\documentclass[11pt,fleqn]{article}
\usepackage{amssymb,latexsym,amsmath,amsfonts,graphicx}
\usepackage[english]{babel}
\usepackage{abstract}
\usepackage{amsmath}
\usepackage{amsfonts}
\usepackage{amssymb}
\usepackage{amsthm}
\usepackage{suffix}
\usepackage{mathtools}
\usepackage{framed}
\usepackage[svgnames]{xcolor}
\usepackage{comment}
\usepackage[textsize=footnotesize]{todonotes}
\usepackage{hyperref} 
\usepackage{pifont}

\usepackage{graphicx}
\usepackage{caption}
\usepackage{subcaption}

\usepackage{tikz}
\usetikzlibrary{shapes,snakes}
\usepackage{pgf}
\usepackage{pgflibraryarrows}
\usepackage{pgflibrarysnakes}
\usetikzlibrary{decorations.text}
\usepgfmodule{shapes}
\usetikzlibrary{decorations.pathmorphing}
\usetikzlibrary{decorations.markings}
\usetikzlibrary{patterns}
\usetikzlibrary{automata}
\usetikzlibrary{positioning}
\usepackage{pgfplots}
\usetikzlibrary{backgrounds}

\usepackage[left=3.5cm, right=3.5cm, top=3cm, bottom=3cm]{geometry}

\tikzset{->-/.style={decoration={
 markings,
 mark=at position #1 with {\arrow{>}}},postaction={decorate}}}
 
 \tikzset{cross/.style={cross out, draw=black, minimum size=5, inner sep=0pt, outer sep=0pt},
cross/.default={1pt}}
 
\usepackage{enumitem}

\definecolor{lightgray}{gray}{1}

\newtheorem{theorem}{Theorem}[section]
\newtheorem{corollary}[theorem]{Corollary}
\newtheorem{proposition}[theorem]{Proposition}
\newtheorem{lemma}[theorem]{Lemma}
\theoremstyle{remark}
\newtheorem{remark}[theorem]{Remark}

\makeatletter
\def\blfootnote{\xdef\@thefnmark{}\@footnotetext}
\makeatother

\DeclarePairedDelimiterX\MeijerM[3]{\lparen}{\rparen}{\begin{matrix}#1 \\ #2\end{matrix}\delimsize\vert\,#3}

\DeclareMathOperator{\Tr}{Tr}

\DeclareMathOperator{\Ai}{Ai}
\DeclareMathOperator{\Real}{Re}
\DeclareMathOperator{\Imag}{Im}
\renewcommand{\Re}{\Real}
\renewcommand{\Im}{\Imag}

\newcommand{\R}{\mathbb{R}}
\newcommand{\C}{\mathbb{C}}

\def\le{\left}
\def\ri{\right}

\def\d{{\rm d}}
\def\1{\operatorname{Id}}
\def\K{\mathbb{K}}

\newcommand{\eins}{\leavevmode\hbox{\small1\kern-3.8pt\normalsize1}}
\DeclareMathOperator{\sgn}{sgn}

\title{Large gap asymptotics at the hard edge for product random matrices and Muttalib-Borodin ensembles}
\author{Tom Claeys\footnote{Institut de Recherche en Math\'ematique et Physique,  Universit\'e
catholique de Louvain, Chemin du Cyclotron 2, B-1348
Louvain-La-Neuve, BELGIUM}, Manuela Girotti\footnotemark[\value{footnote}] \ and Dries Stivigny\footnote{Department of Mathematics, KU Leuven, Celestijnenlaan 200B box 2400, B-3001 Heverlee, BELGIUM}}
\date{}

\numberwithin{equation}{section}

\begin{document}

\maketitle

\renewcommand{\abstractname}{}
\begin{abstract}
\noindent
\textsc{Abstract.} We study the distribution of the smallest eigenvalue for certain classes of positive-definite Hermitian random matrices, in the limit where the size of the matrices becomes large. Their limit distributions can be expressed as Fredholm determinants of integral operators associated to kernels built out of Meijer $G$-functions or Wright's generalized Bessel functions. They generalize in a natural way the hard edge Bessel kernel Fredholm determinant. We express the logarithmic derivatives of the Fredholm determinants identically in terms of a $2\times 2$ Riemann-Hilbert problem, and use this representation to obtain the so-called {\em large gap asymptotics}.
\end{abstract}

\section{Introduction} \label{section: introduction}

It is well known that gap probabilities and extreme eigenvalue distributions of random matrices whose eigenvalues follow a determinantal point process can be expressed as Fredholm determinants corresponding to integral kernel operators.
As the size of the random matrices tends to infinity, universal limit distributions arise, depending on the scaling regime.
Three classical limit distributions for Hermitian random matrices are the Fredholm determinants associated to the sine kernel, the Airy kernel, and the Bessel kernel. Loosely speaking, the sine kernel determinant describes gap probabilities in the bulk of the spectrum, the Airy kernel determinant describes extreme eigenvalue distributions near soft edges, and the Bessel kernel determinant describes extreme eigenvalue distributions near certain hard edges.
Those three determinants are well understood. In particular, they can be expressed identically in terms of Painlev\'e transcendents \cite{JMU, TWFredh, TWAiry, TWBessel} and their asymptotic behaviour for large gaps is known.

\medskip

Recently, there has been a lot of interest in Wishart-type products of random matrices \cite{AkemannIpsenKieburg, AkemannKieburgWei, ArnoDriesMario, ArnoDries, ArnoLun, LiuWangZhang} and in Muttalib-Borodin ensembles \cite{BloomLevenbergWielonsky, Borodin, Cheliotis, TomStefano, ForresterWang, Kuijlaars, Muttalib, LZhang, LZhang2}. 
New universal limiting kernels near the hard edge have been discovered in this context, associated to kernels built out of Meijer $G$-functions \cite{ArnoLun} and Wright's generalized Bessel functions \cite{Borodin}.
The study of the associated Fredholm determinants, which describe the limit distributions of the smallest eigenvalue, was initiated recently in \cite{Strahov} for products of random matrices and in \cite{LZhang2} for Muttalib-Borodin ensembles, and remarkable systems of differential equations have been obtained.
We contribute to these developments by obtaining large gap asymptotics, and by expressing the Fredholm determinants identically in terms of a $2\times 2$ Riemann-Hilbert problem.

\medskip

\paragraph{The random matrix models.}
We are interested in three different types of random matrices, which we describe below.

\begin{enumerate}
\item[(1)] Our first case of interest consists of random matrices $M^{(1)}$ of the form
\begin{equation}
M^{(1)}=\le(G_r\ldots G_{2}G_1\ri)^*G_r\ldots G_{2}G_1,
\end{equation}
where each factor $G_j$ is an independent complex Ginibre matrix of size $(n+\nu_j)\times(n+\nu_{j-1})$, with $\nu_0=0$, and $r\in\mathbb N$, $\nu_1,\ldots, \nu_r\in\mathbb N\cup\{0\}$. This means that all entries of $G_j$ are independent complex standard Gaussians. The notation $^*$ stands for the Hermitian conjugate.
\item[(2)] Our second model consists of products of truncations of Haar distributed unitary matrices. We let $M^{(2)}$ be of the form
\begin{equation}
M^{(2)}=\le(T_r\ldots T_{2}T_1\ri)^*T_r\ldots T_{2}T_1,
\end{equation}
where $r\in\mathbb N$ and $T_j$ is the upper left $(n+\nu_j)\times(n+\nu_{j-1})$ truncation of a Haar distributed unitary matrix $U_j$ of size $\ell_j \times \ell_j$. We assume that $U_1,\ldots, U_r$ are independent, that $\nu_0 = 0$, $\nu_1, ..., \nu_r\in\mathbb N\cup\{0\}$, and that $\ell_j \geq n+\nu_j+1$ for $j = 1, ..., r$. Moreover, we assume that $\sum_{j=1}^r(\ell_j-n-\nu_j)\geq n$.
\item[(3)] Finally, we consider random matrices $M^{(3)}$ whose eigenvalue joint probability distributions take the form
\begin{equation}\label{MB}
\frac{1}{Z_n}\Delta(x_1,\ldots, x_n)\Delta(x_1^\theta,\ldots, x_n^\theta)\prod_{k=1}^n x_k^\alpha e^{-nx_k}\d x_k,
\end{equation}
with $x_1,\ldots, x_n>0$, 
where $\alpha>-1$, $\theta>0$, and
\begin{equation}
\Delta(x_1,\ldots, x_n)=\prod_{1\leq i<k\leq n}(x_k-x_i).
\end{equation}
Such densities are known as Muttalib-Borodin Laguerre ensembles and were shown recently to arise naturally as joint probability densities for the squared singular values of certain upper-triangular matrix ensembles \cite{Cheliotis, ForresterWang}.
\end{enumerate}

The simplest case of models (1) and (3) is the Wishart/Laguerre ensemble. If we set $r=1$ in (1), the matrix $M^{(1)}$ has the form $G^*G$ with $G$ a complex Ginibre matrix of size $(n+\nu_1)\times n$. Such a matrix is called a Wishart/Laguerre random matrix, and the joint probability distribution of the eigenvalues is given by \eqref{MB} in the case $\theta=1$ and $\alpha=\nu_1$. 
The simplest case of model (2), corresponding to $r=1$, is the Jacobi Unitary Ensemble (see \cite[Section 4]{ArnoDriesMario}).

In each of the above models, the joint probability density function for the eigenvalues is a determinantal point process. The associated correlation kernels $K_n^{(j)}$ for the eigenvalues of $M^{(j)}$ have various representations; for us it is convenient that they can be expressed as double contour integrals of the form 
\begin{equation}\label{correlation kernel Kn intro}
 K_n^{(j)}(x,y)=\frac{1}{4\pi^2}\int_{\gamma}\d u\int_{\tilde\gamma}\d v\, \frac{F_n^{(j)}(u)}{F_n^{(j)}(v)}\frac{x^{-u}y^{v-1}}{u-v},  \qquad j=1,2,3
 \end{equation}
for some functions $F_n^{(j)}$ which depend on $j$ and on the parameters $\{\nu_j\}, \{\mu_k\}, \alpha, \theta$ in the model, and for some contours $\gamma$ and $\tilde\gamma$. Such expressions are known for each of the models corresponding to $j=1,2,3$. For convenience of the reader, we  describe the explicit form of $F_n^{(j)}$ and of the shape of the contours $\gamma$ and $\tilde\gamma$ in Appendix \ref{appendix: fredholm}, see Proposition \ref{prop: finite n kernels}, although we will not use this in what follows.
 
In cases (1) and (3), we will be interested in the limit where $n\to +\infty$ for fixed values of the parameters $\nu_1,\ldots, \nu_r$ and $\alpha, \theta$. In case (2), if we let $n\to + \infty$, we also need to let $\ell_1, ..., \ell_r$ go to infinity. For each $j$, we may choose either to let $\ell_j-n$ go to infinity, or to keep $\ell_j-n$ fixed. The large $n$ behaviour of the eigenvalues of $M^{(2)}$ will depend on these choices. We take $J \subseteq \{2,...,r\}$ a subset of indices with cardinality $0 \leq q := |J| < r$, and we let $\ell_1,...,\ell_r$ go to infinity in such a way that
\begin{align}
& \ell_k - n \rightarrow +\infty, &\mbox{ if } k \notin J,\\
&\ell_{k_j} - n = \mu_j>\nu_j, \, \nu_j\in\mathbb N\cup\{0\}, &\mbox{ if }k_j \in J.
\end{align}

\paragraph{Smallest eigenvalue distribution.} 
 
The smallest particle $x^*:=\min_{1\leq m\leq n}x_m$ in the models which we study has a distribution given by
\begin{equation}
{\rm Prob}\left(x^*>s\right)=1+\sum_{m=1}^n \frac{(-1)^m}{m!}\int_{[0,s]^m}\det\left(K_n^{(j)}(x_i,x_\ell)\right)_{i,\ell=1}^m \d x_1\ldots \d x_m.
\end{equation}
Since $K_n^{(j)}$ is of rank $n$, this is the standard series expansion of the Fredholm determinant of the integral operator acting on $L^2(0,s)$ with kernel $K_n^{(j)}$.
We denote this Fredholm determinant by $\det\left(1-\left.K_n^{(j)}\right|_{[0,s]}\right)$.
Near the origin, the eigenvalue correlation kernels $K_n^{(j)}$ admit scaling limits of the following form: for $x,y>0$,
we have
\begin{equation}\label{scaling limit}
\lim_{n\to +\infty}\frac{1}{c_n^{(j)}}K_n^{(j)}\left(\frac{x}{c_n^{(j)}},\frac{y}{c_n^{(j)}}\right)=\mathbb K^{(j)}(x,y),\qquad j=1,2,3,
\end{equation}
where
\begin{equation}\label{def cj}
c_n^{(1)}=n,\qquad c_n^{(2)} = n\prod_{k\notin J} (\ell_k-n),\qquad c_n^{(3)}=n^{\frac{1}{\theta}},
\end{equation}
for some limiting kernels $\mathbb K^{(j)}$. This was proved in \cite{ArnoLun} for $j=1$, in \cite{ArnoDriesMario} for $j=2$, and in \cite{Borodin} for $j=3$. The limiting kernels can be expressed in terms of Meijer $G$-functions for $j=1,2$ and in terms of Wright's generalized Bessel functions for $j=3$; in the special case $\theta = \frac{m}{n} \in \mathbb{Q}$, the Wright's generalized Bessel functions can be expressed as Meijer $G$-functions as well (see \cite{LZhang2}).
Furthermore, if $\theta  \in \mathbb N$ or $\frac{1}{\theta}  \in \mathbb N$, then the limiting kernels  $\mathbb K^{(3)}$ and $\mathbb K^{(1)}$ are related, as it was shown in  \cite{ArnoDries}.
For us, it is important to express the limiting kernels as double contour integrals. As we will show in Proposition \ref{prop: limiting kernels}, they can be expressed as
\begin{equation}
\label{limiting kernels1}\mathbb K^{(j)}(x,y)=\frac{1}{4\pi^2}\int_\gamma\d u\int_{\tilde\gamma}\d v\frac{F^{(j)}(u)}{F^{(j)}(v)}\frac{x^{-u}y^{v-1}}{u-v} ,\qquad j=1,2,3,
\end{equation}
with 
\begin{align}
&F^{(1)}(z)=\frac{\Gamma(z)}{  \prod_{j=1}^r \Gamma\le(1+\nu_j -z\ri)},\label{f1}\\
&F^{(2)}(z)=\frac{\Gamma(z)\prod_{k = 1}^q \Gamma\le(1+\mu_k -z\ri)}{  \prod_{j=1}^r \Gamma\le(1+\nu_j -z\ri)} \quad (q= |J|),\label{f2}\\
&F^{(3)}(z)=\frac{\Gamma(z+\frac{\alpha}{2})}{\Gamma\left(\frac{\frac{\alpha}{2}+1-z}{\theta}\right)}.\label{f3}
\end{align}

\tikzset{->-/.style={decoration={
 markings,
 mark=at position #1 with {\arrow{>}}},postaction={decorate}}}

\begin{figure}
\centering
\scalebox{.9}{
\begin{tikzpicture}[>=stealth]
\path (0,0) coordinate (O);

\draw[dashed, ->] (-4,0) -- (8,0) coordinate (x axis);
\draw[dashed, ->] (.8,-4) -- (.8,4) coordinate (y axis);

\path (-3,-4) coordinate (P);
\path (-3, 4) coordinate (Q);
\path (0,-1) coordinate (P2);
\path (0,1) coordinate (Q2);
\draw[->- = .7] (P) -- (P2);
\draw (Q) -- (Q2);
\node [above] at (0.2,1) {\large $\gamma$};
\draw[->- = .7]  (P2) .. controls + (45:1cm) and + (-45:1cm) .. (Q2);

\path (5,-4) coordinate (A);
\path (5, 4) coordinate (B);
\path (2,-1) coordinate (A2);
\path (2,1) coordinate (B2);
\draw[->- = .7] (A) -- (A2);
\draw (B) -- (B2);
\draw[->- = .7] (A2) .. controls + (135:1cm) and + (225:1cm) .. (B2);
\node [above] at (1.8,1) {\large $\tilde \gamma$};

\draw[fill] (O) circle [radius=0.05];
\node[below] at (-.1,0) {$-\frac{\alpha}{2}$};
\draw[fill] (-1,0) circle [radius=0.05];
\node[above] at (-1,0) {$-\frac{\alpha}{2}-1$};
\draw[fill] (-2,0) circle [radius=0.05];
\node[below] at (-2,0) {$-\frac{\alpha}{2}-2$};
\draw[fill] (-3,0) circle [radius=0.05];
\node[above] at (-3,.2) {\ldots};

\draw (1.2,0) node[cross,rotate=90] {};
\node[below] at (1.2, 0) {$\frac{1}{2}$};

\draw[fill] (2.5,0) circle [radius=0.05];
\node[below] at (2.5, 0) {$\frac{\alpha}{2} + 1$};
\draw[fill] (4,0) circle [radius=0.05];
\node[above] at (4, 0) {$\frac{\alpha}{2} + 1+\theta$};
\draw[fill] (5.5,0) circle [radius=0.05];
\node[below] at (5.5, 0) {$\frac{\alpha}{2} + 1+2\theta$};
\draw[fill] (7,0) circle [radius=0.05];
\node[above] at (7, .2) {\ldots};

\end{tikzpicture}}
\caption{The contours $\gamma, \tilde\gamma$ involved in the double-integral representation of the kernel $\mathbb{K}^{(3)}$, in the case $\alpha>0$, for $\theta>0$; the dots represents the poles and the zeroes of the function $F^{(3)}$. The real value $\frac{1}{2}$ lies always in between $\gamma$ and $\tilde\gamma$.}
\label{Kalphatheta}
\end{figure}
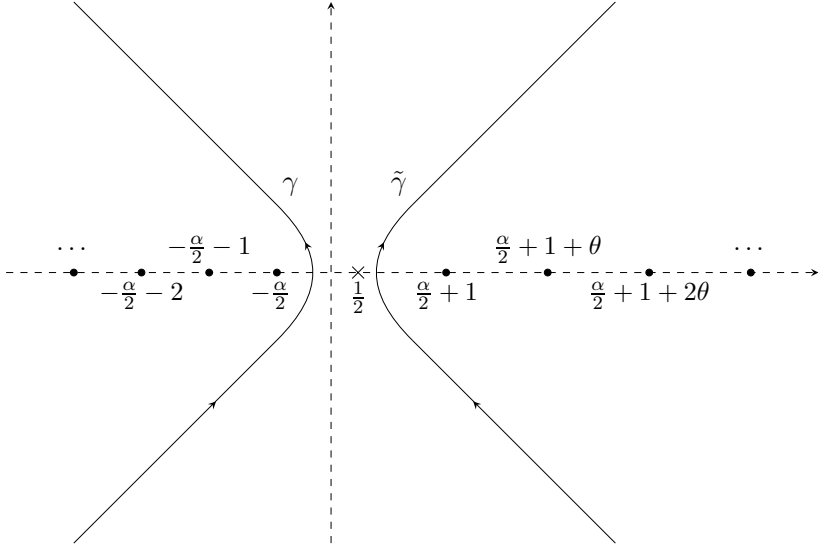
The contours $\gamma$ and $\tilde\gamma$ are such that $\gamma$ lies to the right of the poles of $F^{(j)}$ and $\tilde\gamma$ lies to the left of the zeros of $F^{(j)}$, such that the vertical line through $\frac{1}{2}$ lies in between $\gamma$ and $\tilde\gamma$ (note that this is always possible, since $\mu_k> \nu_k, \, \nu_k\in\mathbb N\cup\{0\}$ and $\alpha>-1$), and such that $\gamma$ and $\tilde\gamma$ do not intersect. Both contours are oriented upwards, and they tend to infinity in sectors lying strictly in the left half plane (for $\gamma$) or the right half plane (for $\tilde\gamma$). The contours are illustrated in Figure \ref{Kalphatheta} for $j=3$, $\alpha>0$ and $\theta>0$.

It is worth noting that the kernels $\mathbb K^{(1)}$ also appear in Cauchy multi-matrix models \cite{BertolaBothner, CauchyBGS, CauchyBGS2}.
The simplest case of all three limiting kernels is the same: if $r=1$, $\nu_1=\alpha$ in case $j=1$, or if $r=1$, $q=0$, $\nu_1=\alpha$ in case $j=2$, or if $\theta=1$ in case $j=3$, we have
\[F^{(1)}\le(z+ \frac{\alpha}{2}\ri)=F^{(2)}\le(z  +\frac{\alpha}{2}\ri)=F^{(3)}(z)=\frac{\Gamma(z+\frac{\alpha}{2})}{\Gamma\left(\frac{\alpha}{2}+1-z\right)},\]
and then the right hand side of \eqref{limiting kernels1} is a well-known (see e.g.\ \cite{LZhang}) integral representation of the hard edge Bessel kernel.

A slightly stronger version of \eqref{scaling limit} allows one to show that 
the large $n$ limit of the smallest particle distribution is given by
\begin{align}\label{lim eig distr Fredholm}
\lim_{n\to +\infty}{\rm Prob}\left(c_n^{(j)}x^*>s\right)=\lim_{n\to +\infty}\det\left(1-\left.K_n^{(j)}\right|_{[0,s/c_n^{(j)}]}\right) \nonumber \\
= \det\left(1-\left.\mathbb K^{(j)}\right|_{[0,s]}\right).
\end{align}
We will justify the last equality in the above formula in Lemma \ref{lemma scaling limit uniform} and Corollary \ref{corollary Fredholm}, using standard results from \cite[Theorem 2.21 and Addendum H]{Simon} about trace class operators.
The Fredholm determinants on the second line of \eqref{lim eig distr Fredholm} are the central objects in what follows.

\paragraph{Fredholm determinants and Riemann-Hilbert problems.}

Let $K$ be an integral operator acting on $L^2(\Sigma)$, with $\Sigma$ be a collection of oriented contours in the complex plane. We call the kernel $K$ of the {\em integrable} form if
\begin{equation*}
K(u, v) = \frac{\textbf{f}(u)^T \textbf{g} (v)}{u - v}, \label{JSkernel2}
\end{equation*}
where $\textbf{f}$ and $\textbf{g}$ are $p$-dimensional column-vectors of sufficiently smooth functions on $\Sigma$, satisfying the condition $\textbf{f}(u)^T \textbf{g}(u) =0$. 

Assume that the kernel $K$ additionally depends on some deformation parameters $\{\kappa_i\}$ and consider its Fredholm determinant $\det (1 - K)$ as a function of such parameters: a well-known procedure due to Its, Izergin, Korepin and Slavnov \cite{IIKS} allows  to express the logarithmic derivatives $\partial_{\kappa_i} \ln \det(1-K)$, for all $ i$, in terms of a Riemann-Hilbert (RH) problem of size $p\times p$. 
For $p=2$, this RH representation is often useful to derive asymptotic properties of the Fredholm determinants by applying the Deift/Zhou steepest descent method on the RH problem. The RH method has lead to rigorous large $s$ asymptotics (which are large gap asymptotics for the underlying determinantal processes) for, amongst others, the sine, Airy, and Bessel kernel Fredholm determinants \cite{BaikBuckinghamdiFranco, DeiftSashaIgor, DeiftItsKrasovskyZhou, DeiftIgorBessel, Krasovsky}.

It is known \cite{ArnoLun, Strahov} that $\mathbb K^{(1)}$ is of integrable form with $p=r+1$, and that $\mathbb K^{(3)}$ is of integrable form if $\theta=a/b\in\mathbb Q$ with $p$ depending on $a$ and $b$ \cite{LZhang2}. Overall, except in the simplest case corresponding to the Bessel kernel, the associated RH problems are of size $p\times p$ with $p>2$ and it is not clear whether such RH problems can be used to obtain large $s$ asymptotics.

A key observation in this paper (see Section \ref{section: identity}), based on ideas from \cite{MeM}, is that the integral operators
$\left.\mathbb K^{(j)}\right|_{[0,s]}$ can be factorized in the form $\mathcal{M}^{-1} \circ \mathbb H_s^{(j)}  \circ \mathcal{M}$ for some suitable operator $\mathcal M$, with $\mathbb H_s^{(j)}$ an integral operator with kernel of integrable form, with $p=2$.
Therefore, we have
\[\det\left(1-\left.\mathbb K^{(j)}\right|_{[0,s]}\right)=\det\left(1-\mathbb H_s^{(j)}\right),\]
and we can express $\frac{\d}{\d s}\ln \det\left(1-\mathbb H_s^{(j)}\right)$ identically in terms of a $2\times 2$ matrix RH problem. The fact that the RH problem is of size $2\times 2$ is remarkable, and it is important to derive large $s$ asymptotics.

We now state the relevant $2\times 2$ RH problem, with jump contours $\gamma$ and $\tilde\gamma$ as defined before.

\subsubsection*{RH problem for $Y$} \begin{itemize}\item[(a)] $Y: \C\setminus(\gamma\cup\tilde\gamma)\to \C^{2\times 2}$ is analytic;
\item[(b)] $Y(z)$ has continuous boundary values $Y_\pm(z)$ as $z$ approaches the contour $\gamma\cup\tilde\gamma$ from the left ($+$) and right ($-$), according to its orientation, and we have the jump relations
\begin{equation}\label{RHP Y: jump}
Y_+ (z)= Y_-(z) J^{(j)}(z),\qquad  z \in  \gamma \cup \tilde \gamma, 
\end{equation}
with jump matrix (the contours are as in \figurename \ \ref{Kalphatheta})
\begin{equation}\label{def Jj}
J^{(j)}(z) = \begin{dcases}
\begin{pmatrix}
1&0\\ s^z F^{(j)}(z)^{-1}&1
\end{pmatrix},& z\in\tilde\gamma,\\
\begin{pmatrix}
1&- s^{-z} F^{(j)}(z) \\0&1
\end{pmatrix},& z\in\gamma,
\end{dcases} 
\end{equation}
where $F^{(1)}, F^{(2)}, F^{(3)}$ are as in \eqref{limiting kernels1};
\item[(c)] as $z\to\infty$, there exists a matrix $Y_1=Y_1(s)$ such that
\begin{equation}\label{as Y}
Y = I + \frac{Y_1(s)}{z}+\mathcal{O}\le(\frac{1}{z^{2}}\ri).
\end{equation}
\end{itemize}
Using the Its-Izergin-Korepin-Slavnov procedure, we will show  that a solution to this RH probem exists; uniqueness of the solution can be shown using standard techniques.
The RH solution $Y$ depends on the value of $j=1,2,3$ and also on $s$ and on the values of the parameters in each of the models, but we will simply write $Y$ for notational convenience.


\paragraph{Statement of results.} As our first result, we establish an identity which expresses the logarithmic $s$-derivative of the Fredholm determinant in terms of the RH solution $Y$.

\begin{theorem}[\bf Differential identity for gap probabilities]\label{theorem identity}
Let $\mathbb K^{(j)}$ be the kernels defined in \eqref{limiting kernels1} for $j=1,2,3$, and let $\det\left(1-\left.\mathbb K^{(j)}\right|_{[0,s]}\right)$ be the Fredholm determinant of the associated operators acting on $[0,s]$ with $s>0$.
Then, we have the identity
\begin{gather}\label{diff id}
\frac{\d }{\d s} \ln \det\left(1-\left.\mathbb K^{(j)}\right|_{[0,s]}\right) =- \frac{1}{s} \le(Y_{1}(s)\ri)_{2,2}, 
\end{gather}
with $\le(Y_{1}(s)\ri)_{2,2} $ the $(2,2)$-entry of $Y_1(s)$, and $Y_1(s)$ defined by \eqref{as Y} in terms of the unique solution to the RH problem for $Y$.
\end{theorem}

We prove the above result in Section \ref{section: identity}.
This RH representation for the Fredholm determinant is particularly useful to study large $s$ asymptotics. 
We will obtain large $s$ asymptotics for $Y$ using the Deift/Zhou steepest descent method and this will enable us to prove the following result.

\begin{theorem}[\bf Large gap asymptotics]
\label{theorem large gap}Let $\mathbb K^{(j)}$ be the kernels defined in \eqref{limiting kernels1} for $j=1,2,3$, and let $\det\left(1-\left.\mathbb K^{(j)}\right|_{[0,s]}\right)$ be the Fredholm determinant of the associated operators acting on $[0,s]$ with $s>0$.
For $j=1,2,3$, there exist constants $C^{(j)}$, $c^{(j)}$ such that, as $s \to +\infty$,
\begin{equation}
\det\left(1-\left.\mathbb K^{(j)}\right|_{[0,s]}\right)= C^{(j)}e^{- a^{(j)}s^{2\rho^{(j)}}+ b^{(j)}s^{\rho^{(j)}} +c^{(j)}\ln s} \le( 1 + o(1)\ri). \label{asymptoticstheta}
\end{equation}
Here, the values of $\rho^{(j)}$ are given by 
\begin{equation}
\label{def rho}
\rho^{(1)}=\frac{1}{r+1},\quad \rho^{(2)}=\frac{1}{r-q+1},\quad \rho^{(3)}=\frac{\theta}{\theta+1}, 
\end{equation}
the values of $a^{(j)}$ by
\begin{align}
&a^{(1)}= \frac{ r^{\frac{1-r}{1+r}} (r+1)^2}{4}      \\
&a^{(2)}= \frac{(r-q)^{\frac{1-r+q}{1+r-q}} (r-q+1)^2}{4}\\
&a^{(3)}= \frac{ \theta^{\frac{1-3\theta}{1+\theta}} (1+\theta)^2}{4},
\end{align}
and the values of $b^{(j)}$ by
\begin{align}
&\hspace{-0.2cm}b^{(1)} = (r + 1) r^{-\frac{r}{r + 1}}  \sum_{j = 1}^r \nu_j  \\
&\hspace{-0.2cm}b^{(2)} = (r -q+ 1) (r-q)^{-\frac{r-q}{r-q + 1}} \left[  \sum_{j = 1}^r \nu_j-\sum_{k = 1}^q \mu_k \right] \\
&\hspace{-0.2cm}b^{(3)} = \frac{\theta+1}{2} \theta^{-\frac{2\theta}{\theta+1}} \left[1+2\alpha-\theta\right],
\end{align}
with $\nu_{\min}=\min \{\nu_1, \ldots, \nu_r\}$.
\end{theorem}

\begin{remark}
We are not able to evaluate the multiplicative constants $C^{(j)}$ explicitly, since they arise as integration constants after integration of the logarithmic derivative \eqref{diff id}. The evaluation of such constants is in general a hard task \cite{Krasovsky}, and our method does not allow to do this. The constants $c^{(j)}$ on the other hand can be computed in principle, but their expressions are involved. We comment on this in Subsection \ref{subsec: log term}.
\end{remark}
\begin{remark}
As stated before, in the case $\theta=1$, the kernel $\mathbb K^{(3)}$ reduces to the hard edge Bessel kernel up to a re-scaling:
\begin{equation}
\mathbb K^{(3)}_{\alpha,\theta=1}(x, y) = 4  \mathbb K_{\operatorname{Bessel},\alpha}(4x,4y), \label{K3andKBessel}
\end{equation}
with
\begin{equation}
\mathbb K_{\operatorname{Bessel},\alpha}(x,y)=  \frac{ J_{\alpha}(\sqrt{x}) \sqrt{y} J_{\alpha}'(\sqrt{y}) - J_{\alpha}(\sqrt{y}) \sqrt{x} J_{\alpha}'(\sqrt{x}) }{ 2 (x-y)}.
\end{equation}
Similarly, if $r=1$, $\nu_1=\alpha\in\mathbb N\cup\{0\}$, and $q=0$, we have
\begin{equation}
\mathbb K^{(1)}_{r=1,\nu_1=\alpha}(x, y) =
\mathbb K^{(2)}_{q=0, r=1, \nu_1=\alpha}(x, y)= 4 \left(\frac{y}{x}\right)^{\alpha/2} \mathbb K_{\operatorname{Bessel},\alpha}(4x,4y).
\end{equation}
Therefore, in these cases, we have
\begin{gather}
 \det\le(1 - \K^{(j)}\bigg|_{[0, s]}\ri)  =  \det\le(1 - \K_{\text{Bessel}}\bigg|_{[0, 4s]}\ri).
\end{gather}
We then have
\[\rho^{(j)}=\frac{1}{2},\quad a^{(j)}=1, \quad b^{(j)}=2\alpha,\] and we obtain
\[\det\left(1-\left.\mathbb K^{(j)}\right|_{[0,s]}\right)= C^{(j)}e^{-s+ 2\alpha \sqrt{s} +c^{(j)}\ln s} \le( 1 + o(1)\ri),\qquad s\to +\infty,\]
which is consistent with 
\cite[formula (9)]{DeiftIgorBessel} and \cite[formula (5)]{Ehrhardt}.
\end{remark}
\begin{remark}
In the case of the product of two Ginibre matrices ($j=1$ and $r=2$), we have
\[\rho^{(1)}=\frac{1}{3},\quad a^{(1)}=\frac{9}{2^{7/3}},\]
and thus
\begin{equation*}
\det\left(1-\left.\mathbb K^{(1)}\right|_{[0,s]}\right)= C^{(1)}e^{-\frac{9}{2^{7/3}}s^{\frac{2}{3}}+ b^{(1)} s^{\frac{1}{3}} +c^{(1)}\ln s} \le( 1 + o(1)\ri),\qquad s\to +\infty,
\end{equation*}
which is consistent with a recent result obtained by Witte and Forrester in \cite[Corollary 3.1]{ForresterWitte}.

Furthermore, in the special case $\nu_1= -\frac{1}{2}$ and $\nu_2=0$, the sub-leading coefficient is equal to $b^{(1)}=-\frac{3}{2^{5/3}}$, which agrees with formula (3.82) from the same paper \cite{ForresterWitte}. 
The same agreement is achieved for $j=3$, $\alpha=0$, $\theta=2$ after the identification $s\mapsto 2\sqrt{s}$ (see \cite[formula (3.79)]{ForresterWitte}), the coefficients being $a^{(3)}=\frac{9}{2^{11/3}}$ and  $b^{(3)}=-\frac{3}{2^{7/3}}$.
\end{remark}


\paragraph{Outline.}
The rest of the paper is organized as follows. In Section \ref{section: identity}, we will use the approach developed by Its-Izergin-Korepin-Slavnov \cite{IIKS} together with ideas from \cite{MeM} to prove the differential identity in Theorem \ref{theorem identity}.
In Section \ref{section: RH}, we then apply the Deift/Zhou steepest descent method to the RH problem for $Y$ to obtain large $s$ asymptotics. This will allow us to prove Theorem \ref{theorem large gap} in Section \ref{sec: fredholm_det}.

\section{Differential identity in terms of RH problem}
\label{section: identity}

\subsection{Some considerations on the limiting kernels}

We note first that \cite[formula (5.11.9)]{NIST} for $x,y\in \R$,
\begin{equation}
\Gamma(x+iy)\sim\sqrt{2\pi}|y|^{x-\frac{1}{2}}e^{-\frac{\pi|y|}{2}},\qquad y\to \pm\infty,
\end{equation}
which implies that, as $y\to\pm\infty$, uniformly in $x$,
\begin{align}
&\label{as F1}F^{(1)}(x+iy)=\le|y\ri|^{(r+1)\le(x-\frac{1}{2}\ri)} e^{\frac{\pi(r-1)}{2}\le|y\ri|}\le|y\ri|^{-\sum_{j=1}^r\nu_j}\mathcal O(1) ,\\
&\label{as F2}F^{(2)}(x+iy)=\le|y\ri|^{(r-q+1)\le(x-\frac{1}{2}\ri)} e^{\frac{\pi(r-q-1)}{2}\le|y\ri|}\le|y\ri|^{\sum_{k=1}^q\mu_k-\sum_{j=1}^r\nu_j}\mathcal O(1) ,\\
&\label{as F3}F^{(3)}(x+iy)= |y|^{\left(1+\frac{1}{\theta}\right) x - \frac{1}{\theta}}e^{\frac{\pi|y|}{2}\left(\frac{1}{\theta}-1 \right)}|y|^{\frac{\alpha}{2}\left(1-\frac{1}{\theta}\right)}\theta^{-\frac{x}{\theta}}\mathcal{O}(1).
\end{align}
As a consequence, it is easily seen that the double integral in \eqref{limiting kernels1} is convergent for our choices of contours $\gamma,\tilde\gamma$ (recall that $\gamma$ lies to the left of the line $\Re z=1/2$, and that $\tilde\gamma$ lies to the right of it).

From the definition of the kernels $\mathbb K^{(j)}$ in \eqref{limiting kernels1} and the functions $F^{(j)}$ in \eqref{f1}-\eqref{f3}, it follows that, for any choice of parameters $\{\nu_j\}, \{\mu_j\}, \alpha, \theta$, there exists $C,\delta>0$ such that for $x,y$ sufficiently small,
\begin{equation}
\left|\mathbb K^{(j)}(x,y)\right|\leq C |xy|^{-\frac{1}{2}+\delta}.
\end{equation}
This implies that, for $j=1,2,3$ and $s>0$,
\begin{equation}\label{norm operators lim}
\int_0^s\int_0^s \left|\mathbb K^{(j)}(x,y)\right|^2 \d x\d y<\infty,
\end{equation}
which means that the integral operators
defined by
\begin{equation}\label{def integral operator}\left.\mathbb K^{(j)}\right|_{[0,s]} f(y)=\int_{0}^s \mathbb K^{(j)}(x,y)f(x)\d x
\end{equation}
are bounded linear operators from $L^2(0,s)$ to itself.

\subsection{Conjugation with Mellin transform}

We will now use the Mellin transform to write the Fredholm determinants in a simpler form. We recall the definition of the Mellin operator and its inverse
\begin{gather}
\mathcal{M}[f](t) = \int_0^\infty x^{t-1} f(x) \d x, \quad \quad 
\mathcal{M}^{-1}[\varphi](x) = \frac{1}{2\pi i}\int_{\frac{1}{2}+i\R} x^{-t} \varphi(t)\d t, 
\end{gather}
which are isometries between $L^2(0,+\infty)$ and $L^2\le( \frac{1}{2}+i\R \ri)$.

\begin{proposition}\label{prop conjugation}
Let $s>0$. For $j=1,2,3$, we have the identity
\begin{equation}\label{identity K H}
\det\left(1-\left.\mathbb K^{(j)}\right|_{[0,s]}\right)=\det\left(1-\mathbb H_s^{(j)}\right),
\end{equation}
where $\mathbb H_s^{(j)}$ is the integral operator acting on $L^2(\tilde\gamma)$ with kernel
\begin{gather}\label{kernel J}
\mathbb H_s^{(j)}(v,z) = \int_\gamma \frac{\d u}{2\pi i} s^{z-u} \frac{F^{(j)}(u)}{F^{(j)}(v)(v-u)(z-u)} .
\end{gather}
\end{proposition}
\begin{proof}
Consider the operator $\left.\mathbb K^{(j)}\right|_{[0,s]}$ acting on $L^2(0,s)$ defined by \eqref{def integral operator}. As an operator on $L^2(0,+\infty)$, its kernel is given by
$\mathbb K^{(j)}(x,y)\chi_{[0,s]}(x)$, where $\chi_{[0,s]}$ is the characteristic function of the interval $[0,s]$.

As a function of $z$, the function \[k(z;s,y):=\int_{\tilde \gamma} \frac{\d v}{2\pi i}  \int_\gamma \frac{\d u}{2\pi i}  \frac{s^{-u}F^{(j)}(u)}{F^{(j)}(v) ( v-u)(z-u)} y^{v  -1}\] is analytic for $z$ in the region at the right of $\gamma$ and for sufficiently large $z$, it can be bounded in absolute value by $C/|z|$ for some $C>0$. For $s\neq x\in \R_+$, we now evaluate the improper integral
\begin{equation}\int_{\frac{1}{2}+i\R}\frac{\d z}{2\pi i}\left(\frac{s}{x}\right)^z k(z;s,y):=\lim_{R\to +\infty}\int_{\frac{1}{2}-iR}^{\frac{1}{2}+iR}\frac{\d z}{2\pi i}\left(\frac{s}{x}\right)^z k(z;s,y).\label{integralk}\end{equation}

First, if $x>s$, by analyticity we can deform the integration contour $\le[\frac{1}{2}-iR,\frac{1}{2}+iR\ri]$ into a semi-circle in the half plane $\Re z>\frac{1}{2}$ with radius $R$.
For $R$ sufficiently large, we obtain in this way,
\[\left|\int_{\frac{1}{2}-iR}^{\frac{1}{2}+iR}\frac{\d z}{2\pi i}\left(\frac{s}{x}\right)^z k(z;s,y)\right|\leq \frac{C}{2\pi }\int_{-\pi/2}^{\pi/2}e^{R\log\left(\frac{s}{x}\right)\cos\theta}\d\theta.\] 
Clearly, the upper bound tends to $0$ as $R\to +\infty$, thus 
\[\int_{\frac{1}{2}+i \R}\frac{\d z}{2\pi i}\left(\frac{s}{x}\right)^z k(z;s,y)=0.\]

Next, if $x<s$, we need to deform the integration contour to the half plane $\Re z< \frac{1}{2}$. For example we can deform $\frac{1}{2}+i\R$ to a contour $\Sigma$ consisting of two straight half-lines starting at $z=\frac{1}{2}$. On such a contour, the $z$-integral in \eqref{integralk} is absolutely convergent, and we can use Fubini's theorem to move the $z$-integral inside the $u$- and $v$-integrals. Using the residue theorem, we then obtain by \eqref{limiting kernels1},
\begin{multline}
\int_{\frac{1}{2}+i \R}\frac{\d z}{2\pi i}\left(\frac{s}{x}\right)^z k(z;s,y)   \\
= \int_{\tilde \gamma} \frac{\d v}{2\pi i}  \int_\gamma \frac{\d u}{2\pi i} \frac{s^{-u}y^{v  -1}F^{(j)}(u)}{F^{(j)}(v) (v-u)}  \int_{\Sigma}\frac{\d z}{2\pi i}\left(\frac{s}{x}\right)^z\frac{1}{z-u} 
=\mathbb K^{(j)}(x,y).
\end{multline}

Thus, both for $x>s$ and $x<s$, we have
\begin{gather}
\mathbb K^{(j)}(x,y)\chi_{[0,s]}(x)
=\int_{\frac{1}{2}+i\R} \frac{\d z}{2\pi i} x^{-z}  \int_{\tilde \gamma} \frac{\d v}{2\pi i}  \int_\gamma \frac{\d u}{2\pi i}  \frac{s^{z-u}F^{(j)}(u)y^{v  -1}}{F^{(j)}(v) ( v-u)(z-u)} .
\end{gather}
Using this triple contour integral representation, we can show that, as an $L^2(0,+\infty)$ operator,  $\left.\mathbb K^{(j)}\right|_{[0,s]}$ can be written as a Mellin conjugation of a simpler integral operator $\mathbb H_s^{(j)}$: we have
\begin{equation}{\mathbb K}^{(j)} \bigg|_{[0,s]} =\mathcal{M}^{-1} \circ \mathbb H_s^{(j)}  \circ \mathcal{M},\label{conjugation}
\end{equation}
where $\mathbb H_s^{(j)}$ is the integral operator acting on $L^2(\tilde\gamma)$ with kernel
$
\mathbb H_s^{(j)}(v,z)$ given by \eqref{kernel J}.
This follows from the following computation on the level of the kernels:
\begin{equation}
\mathcal M\circ {\mathbb K}^{(j)} \bigg|_{[0,s]} (z,y)=s^zk(z;s,y) = \left(\mathbb H_s^{(j)}  \circ \mathcal{M}\right)(z,y).
\end{equation}
It follows from \eqref{conjugation} that \eqref{identity K H} holds,
which completes the proof.
\end{proof}

\subsection{Integrable form}
\begin{proposition}\label{prop J IIKS}
Let $s>0$. For $j=1,2,3$, we have the identity
\begin{gather}\det\left(1-\mathbb H_s^{(j)}\right)= \det \le(1-\mathbb M_s^{(j)} \ri),
\end{gather}
where $\mathbb M_s^{(j)}$ is the integral operator acting on $L^2(\gamma\cup\tilde\gamma)$ 
with kernel 
\begin{equation}\label{integrable form kernel}
\mathbb M_s^{(j)}(u,v) = \frac{\textbf{f}(u)^T  \textbf{g}(v)}{u-v} \end{equation}
where $\textbf{f}$ and $\textbf{g}$ are given by
\begin{equation}\textbf{f}(u) = \frac{1}{2\pi i}\le[ \begin{array}{c} \chi_{\gamma}(u) \\ s^u \chi_{\tilde \gamma} (u) \end{array}\ri] , \quad  \quad
\textbf{g}(v) = \le[\begin{array}{c} -F^{(j)}(v)^{-1} \chi_{\tilde \gamma}(v) \\ s^{-v}F^{(j)}(v) \chi_\gamma(v)  \end{array} \ri],
\end{equation}
with $F^{(1)}$, $F^{(2)}$, $F^{(3)}$ defined by  \eqref{f1}-\eqref{f2}-\eqref{f3}, and 
where $\chi_{\gamma}$ (resp. $\chi_{\tilde \gamma}$) is the characteristic function of the contour $\gamma$ (resp. $\tilde \gamma$).
\end{proposition}

\begin{proof}
The operator $\mathbb H_s^{(j)}$ on $L^2(\tilde\gamma)$ can be written as the composition $A^{(j)}\circ B^{(j)}$ of two operators $A^{(j)}:L^2(\gamma)\to L^2(\tilde\gamma)$ and $B^{(j)}:L^2(\tilde\gamma)\to L^2(\gamma)$, where $A^{(j)}$ and $B^{(j)}$ are integral operators with kernels
\begin{equation}
A^{(j)}(u,z) = \frac{s^{z-u} F^{(j)}(u)}{ 2\pi i (z-u)}, \quad \quad B^{(j)}(v,u) = \frac{1}{F^{(j)}(v)(v-u)}.
\end{equation}
It is then straightforward to check that  $\int_{\tilde\gamma} \int_{\gamma} \le|A^{(j)}\le(\eta,\xi\ri)\ri|^2 |\d \eta| |\d \xi|<+\infty$ and hence $A^{(j)}$ is Hilbert-Schmidt, and similarly for $B^{(j)}$. Therefore, as composition of two Hilbert-Schmidt operators, $\mathbb H_s^{(j)}$ is trace-class.

We now prove that the operators $A^{(j)}$ and $B^{(j)}$ are themselves trace-class. We can write $B^{(j)}=B_2\circ B_1$ as the composition of two Hilbert-Schmidt operators as follows:
\begin{gather}
L^2(\tilde \gamma) \xrightarrow{B_1} L^2\le( {\scriptstyle\frac{1}{2} }+ i\R \ri) \xrightarrow{B_2} L^2(\gamma) \nonumber \\
f(v) \mapsto \int_{\tilde \gamma} \d v \frac{f(v)}{F^{(j)}(v) (v-w)} \mapsto 
\int_{\frac{1}{2}+i\R} \frac{\d w}{2\pi i} \int_{\tilde \gamma}\d v\frac{f(v)}{F^{(j)}(v) (v-w)(u-w)}. \nonumber
\end{gather}

Similarly, we have $A^{(j)}=A_2\circ A_1$ with
\begin{gather}
L^2( \gamma) \xrightarrow{A_1} L^2\le( {\scriptstyle \frac{1}{2} }+ i\R \ri) \xrightarrow{A_2} L^2(\tilde \gamma) \nonumber \\
g(u) \mapsto \int_{ \gamma} \frac{\d u}{2\pi i} \frac{s^{-u}F^{(j)}(u)g(u)}{ w-u}  \mapsto 
\int_{\frac{1}{2}+i\R} \frac{\d w}{2\pi i} \int_{ \gamma} \frac{\d u}{2\pi i} \frac{s^{z-u}F^{(j)}(u)g(u)}{ (w-u)(z-w)}. \nonumber
\end{gather}
It is now easy to verify that the kernels of $A_1, A_2$ and $B_1, B_2$ are Hilbert-Schmidt, hence $A^{(j)}$ and $B^{(j)}$ are trace-class.

As an operator acting on the Hilbert space $L^2(\gamma) \oplus L^2(\tilde\gamma)$, we can write $\mathbb M_s^{(j)}$ as a $2\times 2$ matrix of operators,
\[\mathbb M_s^{(j)}=\le[\begin{array}{c|c} 0&A^{(j)}\\\hline B^{(j)} &0 \end{array}\ri].\]
Moreover, we have the following identities,
\begin{eqnarray*}\det  \le(1-\mathbb H_s^{(j)}\ri)&=&\det \le(1- \le[\begin{array}{c|c}  A^{(j)}\circ B^{(j)}&0\\\hline 0&0 \end{array}\ri]\ri)\\&=&\det \le(1- \le[\begin{array}{c|c}  A^{(j)}\circ B^{(j)}&0\\\hline B^{(j)} &0 \end{array}\ri]\ri)\\&=&
\det \le(1 + \le[\begin{array}{c|c} 0& A^{(j)}\\\hline 0 &0 \end{array}\ri]\ri)\det \le(1 - \le[\begin{array}{c|c} 0&A^{(j)}\\\hline B^{(j)} &0 \end{array}\ri]\ri)
\\&=& \det \le(1 - \mathbb{M}_s^{(j)} \ri),
\end{eqnarray*}
and the result is proved.
\end{proof}

Following the procedure developed by Itz, Izergin, Korepin, and Slavnov \cite{IIKS}, one can relate an integral operator with kernel of the {\em integrable} form \eqref{integrable form kernel} to the  RH problem for $Y$ given in \eqref{RHP Y: jump}-\eqref{as Y}, where we note that the jump matrix $J^{(j)}$ takes the form
\[ J^{(j)}(z) = 
 I - 2\pi i \textbf{f}(z) \textbf{g}(z)^T .
\]
In particular, the resolvent of the operator $\mathbb M_s^{(j)}$ exists if and only if the solution to the above RH problem exists as well, and logarithmic derivatives of the Fredholm determinant $\det(1-\mathbb M_s^{(j)})$ with respect to deformation parameters can be expressed in terms of $Y$. 
In our situation, $s$ plays the role of the deformation parameter, and we can use  the results from \cite[Section 5.1]{Misomonodromic} and \cite{MeM} (proved for general deformation parameters) to conclude that
\begin{gather}
\frac{\d}{\d s} \ln \det \left(1-\mathbb M_s^{(j)}\right) = \int_{\gamma\cup\tilde\gamma} \operatorname{Tr}\left[ Y_-^{-1}(z) Y_-'(z)\, \partial_s J (\lambda) J^{-1}(z) \, \right] \, \frac{\d z}{2\pi i}   \label{Misomonodromictau}
\end{gather}
where $'$ refers to the partial derivative with respect to $z$.
 In the original formula in \cite{MeM}, an additional term is present which depends exclusively on the jumps of the RH problem, but this term is identically zero in our case.
Furthermore, a simple calculation shows  that
\begin{gather}
\int_{\gamma\cup\tilde\gamma} \operatorname{Tr}\left[ Y_-^{-1}(z) Y_-'(z)\, \partial_s J (z) J^{-1}(z) \, \right] \, \frac{\d z}{2\pi i} \nonumber \\
=  \int_{\gamma\cup\tilde\gamma}  \frac{z}{2s}  \le( \operatorname{Tr}\left[ Y^{-1}(z) Y'(z) \sigma_3 \, \right]_+ -  \operatorname{Tr}\left[ Y^{-1}(z) Y'(z) \sigma_3 \, \right]_-\ri) \, \frac{\d z}{2\pi i},
\end{gather}
with $\sigma_3=\begin{pmatrix}1&0\\0&-1\end{pmatrix}$ the third Pauli matrix. Now we use a contour deformation argument to simplify this expression. Using the asympotic behaviour \eqref{as Y} of $Y$ at infinity, one sees that the contribution of the $-$ side of $\gamma$ cancels out with the one of the $+$ side of $\tilde\gamma$. Furthermore, the integrals over the $+$ side of $\gamma$ and the $-$ side of $\tilde\gamma$ can be deformed, and their contribution is equal to the integral of a large clockwise oriented circle. In this way, we obtain that the above expression is equal to
\begin{gather}
= - \lim_{R\to\infty}\frac{1}{2s} \int_{C_R}   \frac{\operatorname{Tr}\left[Y_1 \sigma_3 \, \right] }{z} \, \frac{\d z}{2\pi i} 
=  - \frac{1}{s}  \le(Y_{1}\ri)_{2,2} ,
\end{gather}
where $C_R$ is a counterclockwise oriented circle of radius $R$ around $0$.
The last equality follows because $Y_1$ is traceless, which in turn follows from the fact that $\det Y (z ) = 1$, for all $z \in \C \setminus \le(\gamma \cup \tilde \gamma\ri)$.
This completes the proof of Theorem \ref{theorem identity}.

\begin{remark}
The integral $\int_{\gamma\cup\tilde\gamma} \operatorname{Tr}\left[ Y_-^{-1} Y_-' \partial J J^{-1} \right] \frac{\d z}{2\pi i}$ can be interpreted in terms of the theory of isomonodromic $\tau$-functions \cite{JMU}, as explained in detail in \cite{Misomonodromic}.
\end{remark}

\begin{remark}
The above procedure can be generalized to the case of several intervals instead of the single interval $[0,s]$, similar to \cite{MeM}. 

Let $ J := \bigcup_{j=1}^{m} [a_{2j-1}, a_{2j}]$ be a collection of disjoint intervals on the positive real line: $0 \leq a_1 < a_2 < \ldots < a_{2m}$. Then,
\begin{gather}
\mathbb{K}^{(j)}\bigg|_{J}  = \sum_{k=1}^{2m} (-1)^k \mathbb{K}^{(j)}\bigg|_{[0,a_{k}]},
\end{gather}
and we can still write this operator as a conjugation with a Mellin transform as in \eqref{conjugation}. Then the adaption of the proofs of Propositions \ref{prop conjugation} and \ref{prop J IIKS} is straightforward, and it leads to an integrable operator $\mathbb M_s^{(j)}$ of the form \eqref{integrable form kernel}, but now with $\textbf{f}$ and $\textbf{g}$ of larger dimension. This would lead to a RH problem of larger size, depending on the number of endpoints $2m$.
\end{remark}

\section{Asymptotic analysis of the RH problem}
\label{section: RH}

\label{sgoestoinfty}

The aim of this section is to study the asymptotic behaviour as $s\to +\infty$ of the solution $Y=Y(z ;s)$ to the RH problem stated in Section \ref{section: introduction}.
In particular, we will be interested in the $(2,2)$-entry of the matrix $Y_1$, defined in \eqref{as Y}, which appears in the logarithmic derivative of the Fredholm determinant in each of the cases $j = 1, 2, 3$, as already proved in Section \ref{section: identity}.

To achieve the asymptotic results stated in Theorem \ref{theorem large gap}, we will apply a series of invertible transformations to the RH problem for $Y$ to obtain a RH problem for which we can easily inspect the large $s$ asymptotics of the solution. This procedure is known as the Deift/Zhou steepest descent method \cite{SDM,SDM2}.

\subsection{First transformation $Y \mapsto U$}\label{RHPforU}

In order to simplify the analysis of the RH problem, we define
\begin{equation} \label{eq: transfYT}
U(\zeta) := s^{\frac{\tau^{(j)}}{2}\sigma_3} Y\le(i s^{\rho^{(j)}} \zeta + \tau^{(j)} \ri)s^{-\frac{\tau^{(j)}}{2}\sigma_3},
\end{equation}
with $\tau^{(j)}$ depending on the parameters in the models corresponding to $j = 1, 2, 3$ as follows,
\begin{align*}
& \tau^{(1)} = \frac{\nu_{\min} + 1}{2}, && \rho^{(1)} = \frac{1}{r + 1}, \\
& \tau^{(2)} = \frac{\nu_{\min} + 1}{2}, && \rho^{(2)} = \frac{1}{r - q + 1}, \\
& \tau^{(3)} = \frac{1}{2}, && \rho^{(3)} = \frac{\theta}{\theta + 1},
\end{align*}
where we recall that $\nu_{\min} := \min \{\nu_1, \ldots, \nu_r\}$ and $q := |J|$. From now on, when there is no possible confusion, we omit the $j$-dependence in our notations, and we will simply write $\tau, \rho, F$ instead of $\tau^{(j)}, \rho^{(j)}, F^{(j)}$. The jump contours $\gamma$ and $\tilde\gamma$ are transformed into contours $\gamma_U:=\{\zeta \in \C :is^{\rho^{(j)}}\zeta+\tau^{(j)}\in\gamma\}$ in the upper half plane and $\tilde\gamma_U:=\{\zeta \in \C :is^{\rho^{(j)}}\zeta+\tau^{(j)}\in\tilde\gamma\}$ in the lower half plane, both oriented from left to right. The contours $\gamma_U$ and $\tilde\gamma_{U}$ depend on $s$ and as $s\to +\infty$ they both approach $0$. The poles of the jump matrices $J_U$ now lie on the imaginary axis and accumulate towards the origin in the large-$s$ limit.
$U$ satisfies the following conditions.
\subsubsection*{RH problem for $U$}
\begin{enumerate}
\item $U$ is analytic in $\C \setminus (\gamma_U\cup\tilde\gamma_U)$;
\item $U_+(\zeta) = U_-(\zeta) J_U(\zeta)$ for $\zeta\in\gamma_U\cup\tilde\gamma_U$ with
\begin{equation}\label{def: JT1}
J_U(\zeta) = \begin{dcases}\begin{pmatrix}1&-s^{-i s^{\rho} \zeta} F\le(i s^{\rho} \zeta + \tau\ri)\\0&1 \end{pmatrix} & \text{ if } \zeta \in \gamma_U, \\ \begin{pmatrix}1&0 \\ s^{i s^{\rho} \zeta}F\le(i s^{\rho} \zeta + \tau\ri)^{-1} &1 \end{pmatrix} & \text{ if } \zeta \in \tilde{\gamma}_U;\end{dcases}
\end{equation}
\item $\displaystyle U(\zeta) = I + \frac{U_1(s)}{\zeta} + \mathcal{O}\le(\frac{1}{\zeta^2}\ri)$ as $\zeta \to \infty.$
\end{enumerate}
On the other hand, from \eqref{eq: transfYT} and condition (c) in the RH problem for $Y$, it follows that
\begin{equation*}
U(\zeta) = I + \frac{Y_1(s)}{i s^{\rho}} \zeta^{-1} + \mathcal{O}\le(\zeta^{-2}\ri), \qquad \text{as } \zeta \to \infty
\end{equation*}
which implies the identity
\begin{equation}\label{identity Fredholm U}
\frac{\d}{\d s} \ln \det \le( 1 - \K^{(j)}\bigg|_{[0,s]}\ri) =- i s^{\rho^{(j)} - 1} \le(U_1(s)\ri)_{2,2}. 
\end{equation}

\begin{figure}[h]
\centering
 \scalebox{.9}{
\begin{tikzpicture}[>=stealth]

\draw[dashed, ->] (-5,0) -- (5,0) coordinate (x axis);
\draw[dashed, ->] (0,-4.5) -- (0,4.5) coordinate (y axis);

\begin{scope}[shift={(0,1)}, rotate=-90]
\path (-3,-4) coordinate (P);
\path (-3, 4) coordinate (Q);
\path (0,-1) coordinate (P2);
\path (0,1) coordinate (Q2);
\draw[->- = .7] (P) -- (P2);
\draw (Q) -- (Q2);
\draw[->- = .7]  (P2) .. controls + (45:1cm) and + (-45:1cm) .. (Q2);
\end{scope}
\node [below] at (3.5,3) {\large $\gamma_U$};

\begin{scope}[shift={(0,1)}, rotate=-90]
\path (5,-4) coordinate (A);
\path (5, 4) coordinate (B);
\path (2,-1) coordinate (A2);
\path (2,1) coordinate (B2);
\draw[->- = .7] (A) -- (A2);
\draw (B) -- (B2);
\draw[->- = .7] (A2) .. controls + (135:1cm) and + (225:1cm) .. (B2);
\end{scope}
\node [above] at (3.5,-3) {\large $\tilde \gamma_U$};

\draw[fill] (0,1) circle [radius=0.05];
\draw[fill] (0,2) circle [radius=0.05];
\draw[fill] (0,3) circle [radius=0.05];
\draw[fill] (0,4) circle [radius=0.05];




\draw[fill] (0,-1) circle [radius=0.05];
\draw[fill] (0,-2) circle [radius=0.05];
\draw[fill] (0,-3) circle [radius=0.05];
\draw[fill] (0,-4) circle [radius=0.05];

\end{tikzpicture}}
\caption{The contours in the RH problem for $U$. The poles appearing in the jump matrices lie now on the imaginary axis and accumulate towards the origin as $s \rightarrow + \infty$.}
\label{fig: RHPU}
\end{figure}
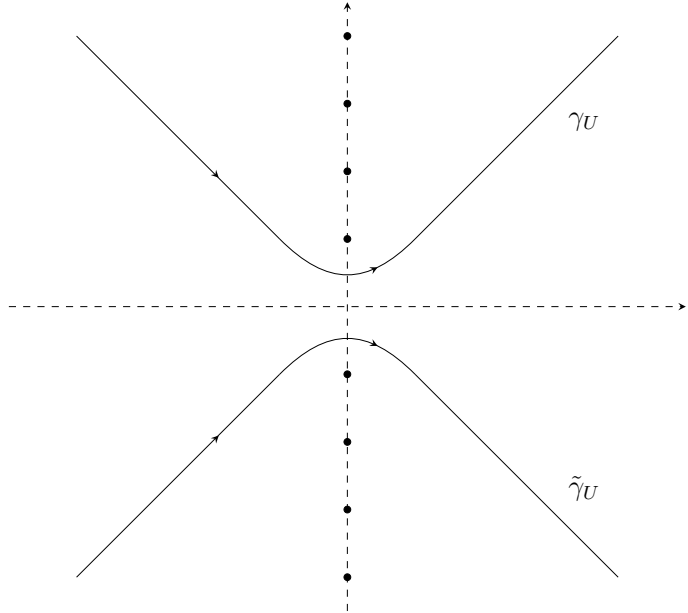

\subsection{Deformation of the contours and transformation $U\mapsto T$}

By analyticity of the jump matrices $J_U$, it follows that the RH solution $U$ can be analytically continued from the region above $\gamma_U$ to $\C\setminus[0,-i\infty)$. We write $U^{\rm I}$ for this analytic extension, which is defined as
\begin{equation}
U^{\rm I}(\zeta)=\begin{cases}U(\zeta),&\mbox{ above $\gamma_U$,}\\
U(\zeta)\begin{pmatrix}1&-s^{-i s^{\rho} \zeta} F\le(i s^{\rho} \zeta + \tau\ri)\\0&1 \end{pmatrix},&\mbox{ below $\gamma_U$.}
\end{cases}
\end{equation} Similarly, we can continue $U$ from the region between $\gamma_U$ and $\tilde\gamma_U$ to $ \C\setminus i\R$, and we denote this function by $U^{\rm II}$.
Finally, we can continue $U$ from the region below $\tilde\gamma_U$ to $\C\setminus[0,+i\infty)$, and we denote this function by $U^{\rm III}$.

Now, we define 
contours $\Sigma_1,\Sigma_2,\ldots, \Sigma_5$ and regions ${\rm I}, {\rm II},{\rm III}, {\rm IV}$ as shown in Figure \ref{fig: TomRHPstraight}, namely
\begin{equation}\label{def Sigma}
\Sigma_5=[b_1,0]\cup [0,b_2], \ \Sigma_2=-\overline{\Sigma_1}=b_2+e^{i(\phi+\epsilon)}(0,+\infty),\ \Sigma_4=-\overline{\Sigma_3}=b_2+e^{-i\epsilon}(0,+\infty),
\end{equation}
with $0<\epsilon<\pi/10$.
The endpoints $b_1 = b_1^{(j)}$ and $b_2=b_2^{(j)}$ are such that $b_2=-\overline{b_1}=be^{i\phi}$ and will be determined later. 
This type of contour will later turn out to be suitable in the steepest descent analysis for case $j=1$, $j=2$, and $j=3$ with $\theta \leq 1$. In what follows, we focus on these cases. If $j=3$ and $\theta > 1$, we need to deform the contours in a different way. We will comment on the changes which have to be made in this case in Remark \ref{remark: phinegative}.
 We define
\begin{equation}\label{RHPforT}
T(\zeta)=\begin{cases}U^{\rm I}(\zeta),&\mbox{ in region I,}\\
U^{\rm II}(\zeta),&\mbox{ in region II and IV,}\\
U^{\rm III}(\zeta),&\mbox{ in region III.}
\end{cases}
\end{equation}

It is straightforward to check that the jump matrices for $T$ on $\Sigma_1,\ldots, \Sigma_4$ are the same as the ones for $U$ on the corresponding contours $\gamma_U$ and $\tilde\gamma_U$. On $\Sigma_5$, the jump matrix for $T$ is obtained by multiplying the jump matrix $J_U$ in \eqref{def: JT1} on $\tilde\gamma$ with the one on $\gamma$. 
%

\begin{figure}[h]
\centering
 \scalebox{1.2}{
\begin{tikzpicture}[>=stealth]

\fill[left color = black!30,right color = white]  (3,1) -- (6,3.5) -- (6,.5) -- (3,1);
\fill[left color = white,right color = black!30]  (-3,1) -- (-6,3.5) -- (-6,.5) -- (-3,1);
\shade[top color=white, bottom color=black!10] (-6,3.5) --  (-3,1) -- (0,0) -- (3,1) -- (6,3.5);
\shade[top color=black!50, bottom color=white] (-6,.5) --  (-3,1) -- (0,0)-- (0,-1.5) -- (-6,-1.5);
\shade[top color=black!50, bottom color=white] (6,.5) --  (3,1) -- (0,0)-- (0,-1.5) -- (6,-1.5);

\draw[dashed, ->] (-6,0) -- (6,0) coordinate (x axis);
\draw[dashed, ->] (0,-1.5) -- (0,4) coordinate (y axis);

\draw[dashed]  (3,1) -- (6,2); 
\draw[dashed]  (-3,1) -- (-6,2);

\draw[->- = .4]  (-3,1) -- (0,0);
\draw[->- = .4]  (0,0) -- (3,1);
\draw[->- = .4]  (-6,.5) -- (-3,1);
\draw[->- = .4]  (3,1) -- (6,.5);
\draw[->- = .4]  (-6,3.5) -- (-3,1);
\draw[->- = .4]  (3,1) -- (6,3.5);

\draw[fill] (-3,1) circle [radius=0.03];
\draw[fill] (3,1) circle [radius=0.03];

\node [above right] at (-3,1) {\scriptsize $b_1$};
\node [above left] at (3,1) {\scriptsize $b_2$};

\node [above left] at (0,.3) {\scriptsize $\Sigma_5$};
\node [left] at (-3,1.9) {\scriptsize $\Sigma_1$};
\node [below left] at (-4,.7) {\scriptsize $\Sigma_3$};
\node [right] at (3,1.9) {\scriptsize $\Sigma_2$};
\node [below right] at (4,.7) {\scriptsize $\Sigma_4$};

\draw (1,0) arc (0:19:1);
\draw (-1,0) arc (180:161:1);


\node [below left] at (1,0) {\scriptsize $\phi$};
\node [below right] at (-1,0) {\scriptsize $\phi$};



\node [above right] at (0,3) { I};
\node [left] at (-4.5,1.9) { II};
\node [below right] at (0,-1) { III};
\node [right] at (4.5,1.9) { IV};

\end{tikzpicture}}
\caption{The contour setting for the RH problem $T(\zeta)$.
}
\label{fig: TomRHPstraight}
\end{figure}
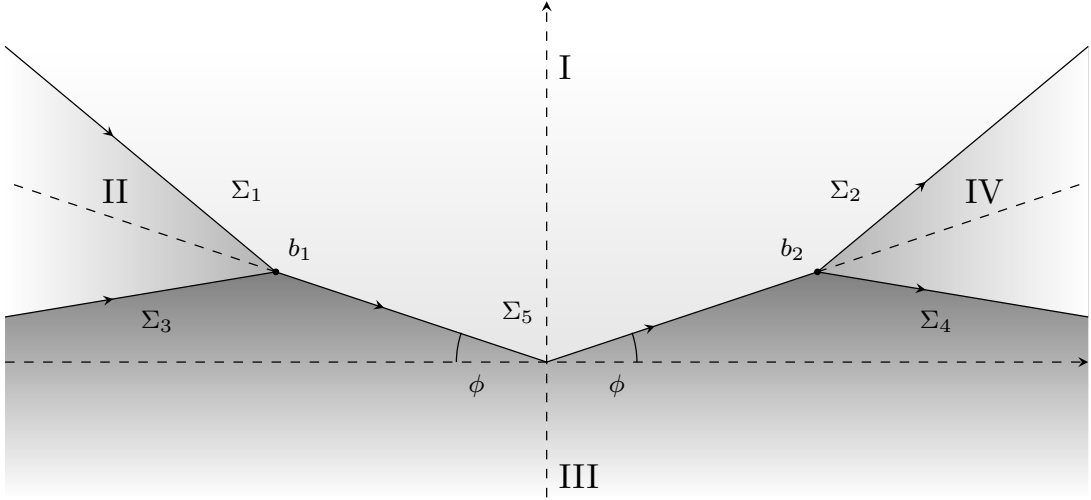

\subsubsection*{RH problem for $T$}
\begin{enumerate}
\item $T$ is analytic in $\C \setminus (\Sigma_1\cup\ldots\cup\Sigma_5)$;
\item $T_+(\zeta) = T_-(\zeta) J_T(\zeta)$ for $\zeta\in\Sigma_1\cup\ldots\cup\Sigma_5$ with
\begin{equation}\label{def: JT1-second}
J_T(\zeta) = \begin{dcases}\begin{pmatrix}1&-s^{-i s^{\rho} \zeta} F\le(i s^{\rho} \zeta + \tau\ri)\\0&1 \end{pmatrix} & \text{ if } \zeta \in \Sigma_1\cup\Sigma_2; \\ \begin{pmatrix}1&0 \\ s^{i s^{\rho} \zeta}F\le(i s^{\rho} \zeta + \tau\ri)^{-1} &1 \end{pmatrix} & \text{ if } \zeta \in \Sigma_3\cup\Sigma_4;\\ \begin{pmatrix}1&-s^{-i s^{\rho} \zeta} F\le(i s^{\rho} \zeta + \tau\ri) \\ s^{i s^{\rho} \zeta}F\le(i s^{\rho} \zeta + \tau\ri)^{-1} &0 \end{pmatrix} & \text{ if } \zeta \in \Sigma_5;\end{dcases}
\end{equation}
\item $\displaystyle T(\zeta) = I + \frac{T_1(s)}{\zeta} + \mathcal{O}\le(\frac{1}{\zeta^2}\ri)$ as $\zeta \to \infty,$
with $T_1(s)=U_1(s)$.
\end{enumerate}

Before proceeding with the next transformation, we will rewrite the jump matrix $J_T$. Recalling Stirling's approximation formula for $z \to \infty$ and $|\text{arg } z| < \pi$,
\begin{equation}
\ln \Gamma (z) = z \ln z - z - \frac{1}{2} \ln z + \frac{1}{2} \ln 2\pi +   \frac{1}{12 z} +   \mathcal{O}\le(\frac{1}{z^3}\ri),
\end{equation}
we obtain from \eqref{f1}--\eqref{f3} that as $s \to + \infty$ with $\zeta$ not too close to $0$ such that $s^{\rho} \zeta \to \infty$, we have 
\begin{multline}
\ln F\le(i s^{\rho}\zeta + \tau \ri) = i s^{\rho} \ln (s) \zeta + is^{\rho}\left[  c_1 \zeta \ln(i \zeta) + c_2 \zeta \ln(-i \zeta) + c_3 \zeta\right] \\
 + c_4 \ln(s)+ c_5  \ln(i \zeta) +c_6  \ln(-i \zeta) + c_7 +  \frac{c_8}{i s^\rho\zeta}  + \mathcal{O}\le(\frac{1}{s^{2\rho}\zeta^2}\ri)\label{as F}, 
\end{multline}
for some constants $\le\{ c_i =c_i^{(j)} \ri\}_{i=1,\ldots, 8}$, depending on the parameters $\{ \nu_j\}, \{\mu_k\}, \alpha, \theta$, and with principal branches of the logarithms. We will use this for $\zeta\in\Sigma_1\cup\ldots\cup\Sigma_5$ and  $s^\rho\zeta$ large. The constants $\{c_i\}_{i = 1, \ldots, 7}$ are given by
\begin{alignat}{3}
&j=1:  && c_1 = 1, && c_2 = r, \nonumber \\
& && c_3 = -(r+1), && \displaystyle c_4 = \frac{\nu_{\text{min}}}{2}  -  \frac{1}{r+1}  \sum_{j=1}^r \nu_j, \nonumber \\
& && \displaystyle c_5 = \frac{\nu_{\text{min}}}{2}, && \displaystyle c_6 = r\frac{\nu_{\text{min}}}{2}  - \sum_{j=1}^r \nu_j,  \nonumber \\
& && \displaystyle c_7 = \frac{1-r}{2}\ln 2\pi,  &&   \displaystyle c_8 = \frac{r+1}{8}\le(\nu_{\rm min}^2-\frac{1}{3}\ri) + \frac{1}{2} \sum_{j=1}^r \nu_j^2 - \frac{\nu_{\rm min}}{2}\sum_{j=1}^r \nu_j; \label{eq: ci_j1} \\ 
& j=2:  && c_1 = 1, && c_2 = r-q, \nonumber \\
& && c_3 = -(r-q+1), && \displaystyle c_4 = \frac{\nu_{\text{min}}}{2} +\frac{1}{r-q+1}\le[ \sum_{k=1}^q \mu_k  - \sum_{j=1}^r \nu_j \ri], \nonumber \displaybreak \\
& && \displaystyle c_5 = \frac{\nu_{\text{min}}}{2}, && \displaystyle c_6 = (r-q)\frac{\nu_{\text{min}}}{2} + \sum_{k=1}^q \mu_k - \sum_{j=1}^r \nu_j,  \nonumber \\
& && \displaystyle c_7 = \frac{1+q-r}{2}\ln 2\pi,  && c_8 = \displaystyle \frac{r-q+1}{8}\le(\nu_{\rm min}^2-\frac{1}{3}\ri) + \frac{1}{2} \le(\sum_{j=1}^r \nu_j^2 -  \sum_{k=1}^q \mu_k^2 \ri) \nonumber \\ 
& && &&\qquad \displaystyle - \frac{\nu_{\rm min}}{2} \le( \sum_{j=1}^r \nu_j - \sum_{k=1}^q \mu_k \ri);  \label{eq: ci_j2} \\
& j=3:  && c_1 = 1, && c_2 = \displaystyle \frac{1}{\theta}, \nonumber \\
& && \displaystyle c_3 = -\frac{\theta+1+\ln \theta}{\theta}, && \displaystyle c_4 = \frac{\theta+(\theta-1)\alpha -1}{2(\theta+1)}, \nonumber \\
& && \displaystyle c_5 = \frac{\alpha}{2}, && \displaystyle c_6 = \frac{\theta-\alpha-1}{2\theta}, \nonumber \\
& && \displaystyle c_7 = -\frac{\theta-\alpha-1}{2\theta}\ln \theta, && \displaystyle c_8 = -\frac{1}{6} + \frac{\alpha}{4} \le(\frac{1}{\theta}-1\ri) + \frac{\alpha^2}{2}\le(\frac{1}{4\theta}+1\ri). \label{eq: ci_j3}
\end{alignat}
The precise values of the constants $c_4$, $c_7$ and $c_8$ are not important for the proof of our results, but will play a role in the evaluation of the coefficient $c^{(j)}$ of the logarithmic term and in further subleading terms in the Fredholm determinant expansion in \eqref{asymptoticstheta}. 

We now define
\begin{equation} \label{eq: def_Gzeta}
G(\zeta)=G^{(j)}(\zeta) := F^{(j)}\le(i s^{\rho^{(j)}} \zeta + \tau^{(j)}\ri) e^{-i s^{\rho^{(j)}} \left(\ln (s) \zeta - h^{(j)}(\zeta)\right)}
\end{equation}
with
\begin{equation} \label{eq: def_hzeta}
h(\zeta)=h^{(j)}(\zeta) := -c_1^{(j)} \zeta \ln(i \zeta) - c_2^{(j)} \zeta \ln(-i \zeta) - c_3^{(j)} \zeta.
\end{equation}
As $s\to +\infty$ and $\zeta$ such that $s^\rho\zeta\to\infty$, we  have by \eqref{as F},
\begin{equation}\label{as G}
\ln G (\zeta)= c_4 \ln s + c_5 \ln \le( i \zeta \ri)+ c_6 \ln \le( -i\zeta\ri)+ c_7 + \frac{c_8}{is^\rho\zeta}  + \mathcal{O}\le(\frac{1}{s^{2\rho}\zeta^2}\ri).
\end{equation}
On the other hand, if $s\to +\infty$ and $\zeta\to 0$ in such a way that $s^\rho\zeta$ is bounded and such that $i s^\rho \zeta + \tau$ is away from the poles of $F$ (see \eqref{f1} - \eqref{f3}), we have 
\begin{gather}\label{lnGbdd}
\ln G(\zeta) = \mathcal{O}\le(1\ri). 
\end{gather}
The jump matrices in the RH problem for $T$ can now be rewritten in terms of $G$ and $h$. We have
\begin{equation}
J_T(\zeta) = 
\begin{dcases}\begin{pmatrix}1&- G(\zeta)e^{-is^{\rho}h(\zeta)}\\0&1 \end{pmatrix}, &\text{ if } \zeta \in \Sigma_1\cup\Sigma_2, \\ \begin{pmatrix}1&0 \\  G(\zeta)^{-1}e^{is^{\rho}h(\zeta)} &1 \end{pmatrix}, & \text{ if } \zeta \in \Sigma_3\cup\Sigma_4,\\
\begin{pmatrix} 1 & - G(\zeta)e^{-is^{\rho} h(\zeta)}  \\  G(\zeta)^{-1} e^{is^{\rho} h(\zeta)}  & 0\end{pmatrix}, & \text{ if } \zeta \in \Sigma_5,\end{dcases}
\end{equation}
with contours $\le\{\Sigma_j\ri\}_{j=1,\ldots,5}$ as before (see \figurename \ \ref{fig: TomRHPstraight}).

\subsection{Third transformation $T \mapsto S$}

We now proceed with a third transformation of the RH problem, where we introduce a $g$-function $g(\zeta) = g^{(j)}(\zeta)$ with specific properties. We define the modified matrix
\begin{equation} \label{eq: transfTS}
S(\zeta) := e^{s^{\rho} \frac{\ell}{2} \sigma_3} T(\zeta) e^{-s^{\rho} \cdot g(\zeta) \sigma_3} e^{-s^{\rho} \frac{\ell}{2} \sigma_3},
\end{equation}
where $\ell=\ell^{(j)} \in \C$ is a constant which is to be determined. 

We would like to construct a $g$-function $g(\zeta)=g^{(j)}(\zeta)$ which satisfies the following properties.

\subsubsection*{Properties for the $g$-function}

\begin{enumerate}
\item
$g$ is analytic in $\C \setminus \Sigma_5$,
\item there exists a constant $\ell$ such that $g$ satisfies the relation
\begin{equation}\label{eq: jump g}
g_+(\zeta) + g_-(\zeta) - i h(\zeta) + \ell=0,\qquad \mbox{for } \zeta\in\Sigma_5,
\end{equation}
\item $g$ has the asymptotics
\begin{equation} \label{eq: cond2_gfun}
g(\zeta) = \frac{g_1}{\zeta}+\mathcal{O}\le(\frac{1}{\zeta^2}\ri), \qquad \text{as } \zeta \to \infty,
\end{equation} 
for some constant $g_1$.
\end{enumerate}

Given such a $g$-function, one verifies that $S$ solves the RH problem below.
\subsubsection*{RH problem for $S$}
\begin{enumerate}
\item $S$ is analytic in $\C \setminus (\Sigma_1 \cup\ldots\cup\Sigma_5)$;
\item for $\zeta\in \Sigma_1 \cup\ldots\cup\Sigma_5$, we have $S_+(\zeta) = S_-(\zeta) J_S(\zeta)$ with 
\begin{equation} \label{eq: jump_S}
J_S(\zeta) = \begin{dcases}\begin{pmatrix}1&- G(\zeta) e^{s^{\rho}(2 g(\zeta) - i h(\zeta) + \ell)}\\0&1 \end{pmatrix}, & \text{ if } \zeta \in \Sigma_1\cup\Sigma_2, \\ \begin{pmatrix}1&0 \\  G(\zeta)^{-1} e^{-s^{\rho}(2 g(\zeta) - i h(\zeta) + \ell)} &1 \end{pmatrix}, & \text{ if } \zeta \in \Sigma_3\cup\Sigma_4, \\ \begin{pmatrix}e^{-s^{\rho}(g_+(\zeta) - g_-(\zeta))} & - G(\zeta)  \\  G(\zeta)^{-1} & 0\end{pmatrix}, & \text{ if } \zeta \in \Sigma_5;\end{dcases}
\end{equation}
\item $\displaystyle S(\zeta) = I + \frac{S_1(s)}{\zeta}+{}\mathcal{O}\le(\frac{1}{\zeta^2}\ri)$ as $\zeta \to \infty$, 
\begin{equation}\label{S1}
\text{with } \left(S_1(s)\right)_{2,2}=\left(U_1(s)\right)_{2,2}+s^\rho g_1.
\end{equation}
\end{enumerate}

\paragraph{Construction of the $g$-function.}

Instead of constructing the $g$-function directly, it turns out to be convenient to inspect its second derivative, and to impose some appropriate constraints to it afterwards. From the properties of $g$, it is clear that $g''$ needs to satisfy the following.

\subsubsection*{Properties for $g''$}
\begin{enumerate}
\item
$g''$ is analytic in $\C \setminus \Sigma_5$,
\item $g''$ satisfies the relation
\begin{equation}\label{eq: jump g2}
g''_+(\zeta) + g''_-(\zeta) = -i\frac{c_1+c_2}{\zeta},\qquad\mbox{for $\zeta\in\Sigma_5$,}
\end{equation}
\item as $\zeta\to\infty$, there is a constant $g_1$ such that
\begin{equation} \label{eq: cond2_gfunsecder}
g''(\zeta) = \frac{2g_1}{\zeta^3}+\mathcal{O}\le(\zeta^{-4}\ri).
\end{equation} 
\end{enumerate}
Given $\Sigma_5$ with endpoints $b_1$ and $b_2=-\overline{b_1}$ (see \figurename \ \ref{fig: TomRHPstraight}), there is a unique function satisfying these properties, and which is such that 
\begin{equation*}
 r(\zeta)g''(\zeta)=\mathcal{O}(1), \quad \text{as } \zeta\to b_1 \text{ and } \zeta \to b_2,
\end{equation*}
where 
\begin{equation}\label{def r}
r(\zeta):=\le[(\zeta-b_1)(\zeta-b_2)\ri]^{\frac{1}{2}},
\end{equation}
and the branch cut is chosen such that $r(\zeta)$ is analytic in $\C\setminus\Sigma_5$ and $r(\zeta)\sim\zeta$ as $\zeta\to\infty$. 
The unique function $g''$ satisfying these properties is given by
\begin{equation}\label{def g2}
g''(\zeta)=-i\frac{c_1+c_2}{2}\left(\frac{1}{\zeta}-\frac{1}{r(\zeta)}+\frac{i \Im b_1}{\zeta r(\zeta)}\right).
\end{equation}

By the asymptotic condition \eqref{eq: cond2_gfunsecder} for $g''$, we can define
\begin{equation}\label{def g g1}
g'(\zeta):=\int_{\infty}^\zeta g''(\xi)\d\xi,\qquad g(\zeta):=\int_\infty^\zeta g'(\xi)\d \xi,
\end{equation}
where the integration contour does not cross $\Sigma_5$. Note that the values of $g'(\zeta)$ and $g(\zeta)$ do not depend on the choice of integration contour.
For arbitrary choices of the endpoints $b_1, b_2$, the function $g$ defined in this way does not satisfy the required properties for $g$. Indeed, for $\zeta \in \Sigma_5$ with $\Re \zeta<0$, we have
\begin{equation}\label{id g1}
g_+'(\zeta)+g_-'(\zeta)=-i(c_1+c_2)\left(\int_{b_1}^\zeta \frac{\d\xi}{\xi}
+\int_\infty^{b_1}\left(\frac{1}{\xi}-\frac{1}{r(\xi)}+\frac{i\Im b_1}{\xi r(\xi)}\right)\d\xi\right).
\end{equation}
Here, the integration from $b_1$ to $\zeta$ can be taken along $\Sigma_5$ and the integration from $\infty$ to $b_1$ along the horizontal half-line from $b_1-\infty$ to $b_1$.
On the other hand, by \eqref{eq: jump g}, we need that
\begin{equation}\label{id g2}
g_+'(\zeta)+g_-'(\zeta)=-ic_1\log(i\zeta)-ic_2\log(-i\zeta)-i(c_1+c_2+c_3),\qquad \zeta\in\Sigma_5.
\end{equation}
Combining \eqref{id g1} and \eqref{id g2}, we obtain after a straightforward calculation the identity
\begin{multline}
-i(c_1+c_2)\left(\log\zeta -\log\frac{\left|\Re b_1\right|}{2}-\frac{i\pi}{2}\sin\phi\right) +i(c_1+c_2)\sin\phi\, {\rm arcsinh}\le[\tan\phi\ri] \\=-ic_1\log(i\zeta)-ic_2\log(-i\zeta)-i(c_1+c_2+c_3),
\end{multline}
where we define $\phi$ by 
\begin{equation}\label{def b1b2}
b_2 = -\overline{b_1}=be^{i\phi},\qquad \phi\in\left(-\frac{\pi}{2},\frac{\pi}{2}\right).
\end{equation}

Equating the real and imaginary parts of this equation, we obtain the value of the endpoints $b_1, b_2$ from the equations
\begin{align}
&\sin\phi =\frac{c_2-c_1}{c_2+c_1},\\
&\Re b_1=-\Re b_2=-2 \le( \frac{c_2}{c_1}\ri)^{-\frac{c_2-c_1}{2(c_2+c_1)}} e^{-\frac{c_1+c_2+c_3}{c_1+c_2}}. \label{eq: reb2}
\end{align}
Here we used the identities $\rm{arcsinh} [\tan(\phi)] = \ln \left(\tan\left(\frac{\phi}{2} + \frac{\pi}{4}\right)\right)$ and $\tan\left(\frac{\phi}{2} + \frac{\pi}{4}\right) = \sec(\phi) + \tan(\phi)$ which are valid for $\phi\in\left(-\frac{\pi}{2},\frac{\pi}{2}\right)$. Note that for $j = 1, 2$ we have that $c_1 + c_2 + c_3 = 0$ such that \eqref{eq: reb2} simplifies to read 
\begin{equation} \label{eq: reb2_j12}
\Re b_1=-\Re b_2 = -2 \left(r-q\right)^{ -\frac{r-q-1}{2(r-q+1)}},
\end{equation}
where $q=0$ for the case $j=1$.
For $j = 3$ we see that $c_1 + c_2 + c_3 = - \frac{\ln(\theta)}{\theta}$ and hence
\begin{equation} \label{eq: reb2_j3}
\Re b_1=-\Re b_2 = -2 \theta^{\frac{3-\theta}{2(1+\theta)}}.
\end{equation}

If now $c_1 + c_2 + c_3 = 0$ and moreover $c_1=c_2$ (this is true in the special case where either $\theta=1$ or $r=1, q=0$; in these cases the limiting kernel $\mathbb K^{(j)}$ is the Bessel kernel), we have $b_1,b_2\in \R$ and $b_1 =-b_2 = -2$.  In the cases we focus on, i.e. for $j=1$, $j=2$, and for $j=3$ with $\theta\leq 1$, we have $c_2>c_1$ and thus $\phi \geq 0$. For $j = 3$ with $\theta > 1$ on the other hand, we have $c_2 < c_1$ and thus $\phi < 0$. We return to this matter in Remark \ref{remark: phinegative}.

Finally, the constant $\ell$ can now be defined as $\ell := i h(b_1) - 2 g(b_1)$, i.e.
\begin{equation} \label{def: ell constant}
\ell = i h(b_1) - 2 \int_\infty^{b_1} g'(\xi) \d \xi.
\end{equation} 

We now have the following lemma.

\begin{lemma} \label{lemma3.1}
Let $b_1, b_2$ be given by \eqref{def b1b2} and suppose that $\phi\geq 0$.
Let $\Sigma_1,\ldots, \Sigma_5$ be as in \eqref{def Sigma}, see also Figure \ref{fig: TomRHPstraight}. Then the following inequalities hold:
\begin{align}
&\Re \le[ g_+(\zeta) - g_-(\zeta)\ri] > 0,&& \zeta \in \Sigma_5 \setminus \le\{b_1,b_2 \ri\}, \label{eq: cond4_gfun} \\
& \Re\le[2g(\zeta) - i h(\zeta) - \ell \ri] < 0,  && \zeta \in \Sigma_1\cup\Sigma_2, \label{eq: cond5_gfun} \\
&  \Re\le[2g(\zeta) - i h(\zeta) - \ell \ri] > 0,&& \zeta \in \Sigma_3\cup\Sigma_4. \label{eq: cond6_gfun}
\end{align}
%
\end{lemma}

\begin{proof}
For $\zeta \in \Sigma_5$, define 
\begin{equation}
\varphi(\zeta) := g_+ (\zeta) - g_-(\zeta).
\end{equation}
Since
\begin{gather}
g''(\zeta) = -i\frac{c_1+c_2}{2} \le( \frac{1}{\zeta} - \frac{1}{r(\zeta)} + \frac{i \Im(b_1)}{\zeta \, r(\zeta)} \ri) \quad  \text{with } r(\zeta) = \le[(\zeta-b_1)(\zeta-b_2)\ri]^{\frac{1}{2}},
\end{gather}
we have
\begin{gather}
\varphi''(\zeta) = g_+''(\zeta) - g_-''(\zeta) =  i (c_1+c_2)  \frac{\zeta - i \Im(b_1)}{\zeta \, r_+(\zeta)}.
\end{gather}

By the symmetry of the RH problem and the $g$-function, it is sufficient to prove (\ref{eq: cond4_gfun}) for $\zeta \in (b_1,0) \subseteq \Sigma_5$. First of all, we notice that
\begin{gather}
\Re \le[\varphi(\zeta)\ri] = \Re \le[ \int_{b_1}^\zeta \int_{b_1}^\xi \varphi''(\eta) \, \d \eta \, \d \xi\ri]. 
\end{gather} 
Therefore, in order to get \eqref{eq: cond4_gfun}, we only need to prove that $\arg\le[\varphi''(\zeta)\, \d \eta\,\d \xi\ri]$ belongs to the right half plane,
\[\arg\le[\varphi''(\zeta)\, \d\eta\,\d\xi\ri]=\arg \le[ i (c_1+c_2)  \frac{\zeta - i \Im(b_1)}{\zeta \, (\zeta-b_1)_+^{1/2}(\zeta-b_2)_+^{1/2}}\, \d\eta\,\d\xi\ri]\in \le(-\frac{\pi}{2}, \frac{\pi}{2}\ri).\] This is easily achieved, since 
\begin{align*}
&\arg\le[ i(c_1+c_2)\ri]=\frac{\pi}{2},  &&\arg\le[ \d\eta\,\d\xi\ri]=-2\phi, \\
&\arg \le[ \frac{1}{\zeta}\ri]=-\pi+\phi, && \arg\le[ (\zeta-b_1)_+^{-1/2}\ri]=\frac{\phi}{2},\\ 
&\arg\le[ (\zeta-b_2)_+^{-1/2} \ri] \in \le(-\frac{\pi+\phi}{2}, -\frac{\pi}{2} \ri), &&\arg\le[\zeta - i\Im(b_1)\ri]\in \le( -\pi, -\frac{\pi}{2}\ri), 
\end{align*} 
which implies that 
\[\arg\le[\varphi''(\zeta)\, \d\eta\,\d\xi\ri]\in \left(-\phi, \frac{\pi}{2}-\frac{\phi}{2}\right), \quad \text{with } \phi \in \le[0, \frac{\pi}{2} \ri).\]
%
%
%
%
%
%
%
%
%
%
%
%

Next, we want to show that the quantity $\Re \le[ 2g(\zeta) -ih(\zeta) + \ell  \ri]$ is positive on $\Sigma_3\cup\Sigma_4$ and negative on $\Sigma_1\cup\Sigma_2$. We focus again only on the parts of the contours lying in the left half of the complex plane. 

By construction of the $g$-function, we have $g_+(\zeta) + g_-(\zeta) -ih(\zeta) + \ell =0$ on $\Sigma_5$, hence
\begin{equation}
2g_+(\zeta) -ih(\zeta) + \ell = g_+(\zeta) - g_-(\zeta) = \varphi(\zeta) \qquad  \text{on }\Sigma_5.
\end{equation} 
This implies that $2g -ih + \ell $ is the analytic continuation of the function $\varphi$ to the positive side of the curve $\Sigma_5$, i.e. the region above $\Sigma_5$. 

The second derivative of $2g - i h + \ell$ is given by
\begin{equation}
2g''(\zeta) - i h''(\zeta) = i (c_1+c_2)  \frac{\zeta - i \Im(b_1)}{\zeta \, \check{r}(\zeta)} 
\end{equation}
with $\check{r}(\zeta) = \le[ (b_1-\zeta)(b_2-\zeta ) \ri]^{\frac{1}{2}}$ such that $\check{r}$ is analytic on $\C \setminus \left\{ (b_1, b_1 - i \infty) \cup (b_2, b_2 - i \infty) \right\}$ and $\check{r}(\zeta) \in i\R_+$ on the horizontal segment $(b_1,b_2)$. 

For $\zeta \in\Sigma_1$, we have
\begin{gather}
\Re \le[\varphi(\zeta)\ri] = \Re \le[ \int_\zeta^{b_1} \int_\xi^{b_1}\left(2g''(\eta) - ih''(\eta) \right) \, \d \eta \, \d \xi\ri], 
\end{gather}
and it suffices to show that $\arg \le[\left(2g''(\eta) - ih''(\eta)\right) \, \d \eta \, \d \xi\ri]$ lies in the left half plane, meaning 
$$\arg \le[\left(2g''(\eta) - ih''(\eta)\right) \, \d \eta \, \d \xi\ri] = \arg \le[ i (c_1+c_2)  \frac{\zeta - i \Im(b_1)}{\zeta \, \check{r}(\zeta)} \, \d \eta \, \d \xi\ri]  \in \le(\frac{\pi}{2}, \frac{3\pi}{2}\ri).$$
This follows from 
\begin{align*}
&\arg \le[ i(c_1+c_2)\ri]=\frac{\pi}{2},  &&\arg\le[ \d\eta\,\d\xi\ri]=2\pi -2\phi-2\epsilon, \\
&\arg \le[ \frac{1}{\zeta}\ri] \in \le( -\pi+\phi,-\pi+\phi+\epsilon \ri), && \arg\le[(\zeta-b_1)^{-1/2}\ri]=-\frac{\pi-\phi-\epsilon}{2},\\ 
&\arg\le[(\zeta-b_2)^{-1/2}\ri] \in \le(-\frac{\pi}{2}, -\frac{1}{2}\arg \le[\zeta - i\Im(b_1)\ri] \ri), &&\arg\le[ \zeta - i\Im(b_1)\ri] \in \le(\pi-\phi-\epsilon,\pi\ri),
\end{align*} 
which implies that the argument lies in
\[\left(\frac{3\pi}{2}-\frac{3\phi}{2}-\frac{5\epsilon}{2}, \frac{3\pi}{2}-\frac{\phi}{2}-\frac{\epsilon}{2}\right)\subset\left(\frac{\pi}{2}, \frac{3\pi}{2}\right)\]
for $0 \leq \phi < \pi/2$ and $0< \epsilon<\pi/10$.

Finally, in order to prove that $\Re \le[ 2g(\zeta) -ih(\zeta) + \ell  \ri] > 0$ on $\Sigma_3$, we need to show that 
$$\arg \le[i (c_1+c_2)  \frac{\zeta - i \Im(b_1)}{\zeta \, (\zeta-b_1)^{1/2}(\zeta-b_2)^{1/2}} \, \d\eta\,\d\xi\ri] \in \le( -\frac{\pi}{2}, \frac{\pi}{2}\ri),$$
and this follows from 
\begin{align*}
&\arg \le[ i(c_1+c_2)\ri] =\frac{\pi}{2},  &&\arg\le[\d\eta\, \d\xi\ri]= 2\pi + 2\epsilon, \\
&\arg \le[\frac{1}{\zeta}\ri] \in \le(-\pi -\epsilon,-\pi+\phi\ri), && \arg \le[(\zeta-b_1)^{-1/2} \ri] =  -\frac{\pi + \epsilon}{2},\\ 
&\arg\le[ (\zeta-b_2)^{-1/2}\ri] \in \le( -\frac{1}{2}\arg\le[\zeta - i \Im(b_1) \ri], - \frac{\pi}{2} \ri), &&\arg\le[ \zeta - i\Im(b_1)\ri] \in \le( \pi, \pi + \epsilon\ri), 
\end{align*} 
implying that the argument belongs to the interval
\begin{equation*}
\le( - \frac{\pi}{2}, -\frac{\pi}{2} + \phi + \frac{5}{2}\epsilon \ri)\subset\left(-\frac{\pi}{2}, \frac{\pi}{2}\right)
\end{equation*}
for $0 \leq \phi < \pi/2$ and $0< \epsilon<\pi/10$.
\end{proof}

As a consequence of the above lemma, we obtain that the jump matrices $J_S$ for $S$ converge exponentially fast, as $s\to +\infty$, to the identity matrix on $\Sigma_1, \Sigma_2, \Sigma_3, \Sigma_4$, and that the diagonal of $J_S$ converges to $0$ exponentially fast on $\Sigma_5$. This convergence is however not uniformly valid near the endpoints $b_1$ and $b_2$.

\subsection{The global parametrix}
We look for an approximation to $S$ that is valid for large $s$ away from the endpoints $b_1, b_2$. To that end, we want to find a matrix-valued function $P^{\infty}(\zeta)$
satisfying the following RH conditions.

\subsubsection*{RH problem for $P^{\infty}$}
\begin{enumerate}
\item ${P}^{\infty}$ is analytic in $\C \setminus \Sigma_5$;
\item $P^{\infty}_+(\zeta) = P^{\infty}_-(\zeta) J^{\infty}(\zeta)$ with 
\begin{equation} \label{eq: jump_pinfty}
J^{\infty}(\zeta) = \begin{pmatrix}0 & - G(\zeta) \\  G(\zeta)^{-1} & 0\end{pmatrix},\qquad \zeta\in\Sigma_5;
\end{equation}
\item as $\zeta\to\infty$, we have
\begin{equation}{P}^{\infty}(\zeta) = I + \frac{P^\infty_1(s)}{\zeta} + \mathcal{O}\le(\frac{1}{\zeta^2}\ri). \label{Pinftyinfty}
\end{equation}
\end{enumerate}
In order to construct the solution, we first solve a similar and simpler RH problem with constant jumps.

\subsubsection*{RH problem for $Q^\infty$}
\begin{enumerate}
\item $Q^{\infty}$ is analytic in $\C \setminus \Sigma_5$;
\item $\displaystyle Q^{\infty}_+(\zeta) = Q^{\infty}_-(\zeta) \begin{pmatrix}0 & -1 \\ 1 & 0\end{pmatrix}$ for $\zeta\in\Sigma_5$;
\item $\displaystyle Q^{\infty}(\zeta) = I + \mathcal{O}\le(\frac{1}{\zeta}\ri)$ as $\zeta \to \infty$.
\end{enumerate}
The solution to this RH problem is explicit and it is given by (see \cite[Chapter 7]{DeiftCourant} for a similar construction)
\begin{equation}
Q^{\infty}(\zeta) = \frac{1}{2} \begin{pmatrix}\gamma(\zeta)+\gamma(\zeta)^{-1} & \frac{1}{i} \left(\gamma(\zeta)-\gamma(\zeta)^{-1}\right) \\-\frac{1}{i} \left(\gamma(\zeta)-\gamma(\zeta)^{-1}\right) & \gamma(\zeta)+\gamma(\zeta)^{-1}\end{pmatrix},
\end{equation}
where 
\begin{equation}
\gamma(\zeta)=\left(\frac{\zeta-b_1}{\zeta - b_2}\right)^{1/4}
\end{equation}
which is defined and analytic on $\C \setminus \Sigma_5$, with branch cut on $\Sigma_5$. 

\subsubsection*{Construction of $P^\infty$}
Now, $P^\infty$ can be written in the form
\begin{equation} \label{eq: Pinfty}
P^{\infty}(\zeta) :=e^{-p_0 \sigma_3} Q^{\infty}(\zeta) e^{p(\zeta) \sigma_3} 
\end{equation}
with $p(\zeta)=p^{(j)}(\zeta)$ and $p_0=p_0^{(j)}$ suitably defined as 
\begin{align} 
p(\zeta) &= -\frac{r(\zeta)}{2\pi i} \int_{\Sigma_5} \frac{\ln  G(\xi)}{r_+(\xi)} \frac{\d \xi}{\xi-\zeta} \label{eq: pfunction}\\
 p(\zeta) &= p_0 + \frac{p_1}{\zeta} + \mathcal{O}\le(\frac{1}{\zeta^2}\ri) \qquad \text{as } \zeta \rightarrow \infty  ,
\end{align}
with $r(\zeta)$ as in \eqref{def r}, such that $p(\zeta)$ satisfies
\begin{equation}\label{jump p}
p_+(\zeta) + p_-(\zeta) = -\ln  G(\zeta),
\end{equation}
for $\zeta\in\Sigma_5$. After expanding (\ref{eq: pfunction}) at $\zeta = \infty$, we identify
\begin{align}
p_0 &= \frac{1}{2\pi i} \int_{\Sigma_5} \frac{\ln  G(\xi)}{r_+(\xi)} \d \xi, \\
p_1 &= - \frac{b_1+b_2}{4\pi i} \int_{\Sigma_5} \frac{\ln  G(\xi)}{r_+(\xi)} \d \xi + \frac{1}{2\pi i} \int_{\Sigma_5} \frac{\xi \ln  G(\xi)}{r_+(\xi)} \d\xi \nonumber \\
&= \frac{1}{2\pi i} \int_{\Sigma_5} \frac{(\xi-i\Im (b_1) )\ln  G(\xi)}{r_+(\xi)} \d\xi. \label{def p1}
\end{align}
As $\zeta\to b_k$, the solution $P^\infty$ behaves like
\begin{equation*}
P^{\infty}(\zeta) =\mathcal{O}\le((\zeta -b_k)^{-\frac{1}{4}}\ri),\qquad k=1,2.
\end{equation*}

\subsection{The local parametrix at the endpoints $b_1,b_2$}

Near the endpoints $b_1, b_2$ the global parametrix $P^{\infty}$ cannot be a good approximation for $S$, since it blows up, while $S$ remains bounded. Hence, we need to introduce local parametrices near these points. 
The local parametrix $P$ will be defined in small neighbourhoods of $b_1$ and $b_2$,
\begin{equation*}
\mathbb D_{\delta}(b_1) = \{z \in \C : |z - b_1| < \delta\},\qquad \mathbb D_{\delta}(b_2) = \{z \in \C : |z - b_2| < \delta\},
\end{equation*}
for some small but fixed $\delta > 0$, independent of $s$.
We will focus on the parametrix $P$ near the endpoint $b_1$. By symmetry, the parametrix near $b_2$ will be given by $P(\zeta)=\overline{P(-\overline{\zeta})}$.  We will construct $P$ in such a way that it has the same jumps as $S$ in $\mathbb D_{\delta}(b_1)$ and such that it matches with the global parametrix $P^{\infty}$ on the circle $\partial \mathbb D_{\delta}(b_1)$. \\ 

The RH problem that we require $P$ to satisfy is the following (see also Figure \ref{figure: Airy}):

\subsubsection*{RH problem for $P$}
\begin{enumerate}
\item $P$ is analytic in $\mathbb D_{\delta}(b_1) \setminus \left(\Sigma_5 \cup \Sigma_1\cup\Sigma_3\right)$;
\item $P_+(\zeta) = P_-(\zeta) \begin{dcases}\begin{pmatrix}1 & - G(\zeta) e^{s^{\rho} \le( 2g(\zeta) - ih(\zeta) +\ell\ri)} \\ 0 & 1\end{pmatrix}, & z \in \Sigma_1 \cap \mathbb D_{\delta}(b_1) ,\\ \begin{pmatrix}1 & 0 \\  G(\zeta)^{-1} e^{-s^{\rho} \le(2g(\zeta)- ih(\zeta) + \ell\ri)} & 1\end{pmatrix}, & z \in \Sigma_3 \cap \mathbb D_{\delta}(b_1), \\ \begin{pmatrix}e^{-s^{\rho} \le(g_+(\zeta) - g_-(\zeta) \ri)} & - G(\zeta) \\  G(\zeta)^{-1} & 0\end{pmatrix}, & z \in \Sigma_5 \cap \mathbb D_{\delta}(b_1);\end{dcases}$
\item $\displaystyle P(\zeta) =  P^{\infty}(\zeta)\le(I + o(1)\ri)$ as $s\to + \infty$ for $\zeta \in \partial \mathbb D_{\delta}(b_1)$.
\end{enumerate}

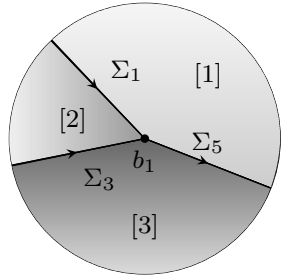
\begin{figure}[h!]
\centering
 \scalebox{1.2}{
\begin{tikzpicture}[>=stealth]

\draw[clip] (0,0) circle (1.5cm);

\filldraw[bottom color=black!20, top color=white] (0,0) -- (1.4,-.55) arc (-20:128:1.8cm) -- cycle;
\filldraw[top color=black!60, bottom color=white] (0,0) -- (1.4,-.55) arc (-20:-164.5:1.8cm) -- cycle;
\filldraw[right color=black!40, left color=white] (0,0) -- (-1.05,1.1) arc (128:181:1.8cm) -- cycle;

\draw[->- = .5] (0,0) -- (1.4,-.55);
\draw[->- = .5] (-1.05,1.1) -- (0,0);
\draw[->- = .5] (-1.5,-.3) -- (0,0);

\node [below] at (0,0) {\scriptsize $b_1$};

\node [above] at (-.2,.5) {\scriptsize $\Sigma_1$};
\node [below] at (-.5,-.2) {\scriptsize $\Sigma_3$};
\node [above] at (.7,-.3) {\scriptsize $\Sigma_5$};

\node at (.7,.7) {\scriptsize $[1]$};
\node at (-.8,.2) {\scriptsize $[2]$};
\node [below] at (0,-.7) {\scriptsize $[3]$};

\draw[fill] (0,0) circle [radius=0.04];

\end{tikzpicture}}
\caption{The jump contours for the local parametrix around the endpoint $b_1$.}
\label{figure: Airy}
\end{figure}

In the next paragraph, we construct $P$ explicitly in terms of the Airy function.

\paragraph{The Airy model RH problem} 

We need a slight variation of the standard model RH problem associated to the Airy function which was used for instance in \cite{DeiftCourant, DKMcLVZ2, DKMcLVZ, DZAiry}. 
Therefore, we follow \cite[Section 3.5.1]{ClaeysGrava} and define
\[
    y_\ell(\zeta)=e^{\frac{2\pi i \ell}{3}}\Ai(e^{\frac{2\pi i \ell}{3}}\zeta),\qquad \ell=0,1,2,
\]
 where $\Ai$ is the Airy function.
Let $A_1$, $A_2$, $A_3$ be entire functions given by
\begin{align}
& A_1(\zeta)=-i\sqrt{2\pi}
        \begin{pmatrix}
            -y_2(\zeta) & -y_0(\zeta)\\
            -y_2'(\zeta) & -y_0'(\zeta)
        \end{pmatrix},\\
& A_2(\zeta)=-i\sqrt{2\pi}
        \begin{pmatrix}
            -y_2(\zeta) & y_1(\zeta)\\
            -y_2'(\zeta) & y_1'(\zeta)
        \end{pmatrix},
        \\
& A_3(\zeta)=-i\sqrt{2\pi}
        \begin{pmatrix}
            y_0(\zeta) & y_1(\zeta)\\
            y_0'(\zeta) & y_1'(\zeta)
        \end{pmatrix}.
\end{align}
Using the well-known Airy function identity
$y_0+y_1+y_2=0$, one verifies the relations
\begin{align}
            \label{RHP A: b1}
            & A_1(\zeta)=A_2(\zeta)
                \begin{pmatrix}
                    1 & -1 \\
                    0 & 1
                \end{pmatrix},
            \\
            & A_2(\zeta)=A_3(\zeta)
                \begin{pmatrix}
                    1 & 0 \\
                    1 & 1
                \end{pmatrix},
            \\
            \label{RHP A: b3}
            & A_1(\zeta)=A_3(\zeta)
                \begin{pmatrix}
                    1 & -1 \\
                    1 & 0
                \end{pmatrix}.
        \end{align}
Moreover, by the asymptotic behaviour for the Airy function in the complex plane, we have that
\begin{align}\label{RHP:A-c}
            A_k(\zeta) &= \zeta^{-\frac{\sigma_3}{4}}\begin{pmatrix}1&i\\1&-i\end{pmatrix}
                \left[I+O\left(\zeta^{-3/2}\right)\right]
                e^{-\frac{2}{3}\zeta^{3/2}\sigma_3}
        \end{align}
        as $\zeta\to\infty$ in the sector $S_k$ for $k=1, 2, 3,$ with
\begin{equation}\label{def sector Sk}
S_k=\le\{\zeta\in\C: \frac{2k-3}{3}\pi+\delta\leq\arg\zeta\leq\frac{2k+1}{3}\pi-\delta \ri\}, \qquad k=1, 2, 3,
\end{equation}
for any $\delta>0$.

\paragraph{Construction of $P$}

We define the local parametrix in the following form,
\begin{equation}\label{def P}
P(\zeta) = E(\zeta) A_k\left(s^{\frac{2}{3} \rho}f(\zeta)\right) e^{s^{\rho} q(\zeta) \sigma_3}  G(\zeta)^{-\frac{\sigma_3}{2}}, 
\end{equation}
for $\zeta$ in region $[k]$ as shown in Figure \ref{figure: Airy}: region $[1]$ is the region between $\Sigma_5$ and $\Sigma_1$; region $[2]$ is the one between $\Sigma_1$ and $\Sigma_3$; region $[3]$ is the one between $\Sigma_3$ and $\Sigma_5$.
Here, $E$ will be an analytic function in $\mathbb D_{\delta}(b_1)$, $q(\zeta)$ is an analytic function on $\mathbb D_{\delta}(b_1) \setminus \Sigma_5$ given by
\begin{equation}
q(\zeta) := g(\zeta) - \frac{i}{2} h(\zeta) + \frac{\ell}{2},
\end{equation} and $f(\zeta)$ will be a conformal map from $\mathbb D_{\delta}(b_1)$ to a neighborhood of $0$ which we will determine below. 

Then, by the form of the jump matrices for $P$ and by the properties of $g$, it is straightforward to  verify that conditions (a) and (b) of the RH problem for $P$ are satisfied.
In order to achieve the matching condition (c) as well, we need to define the analytic prefactor $E$ and the conformal map $f$ appropriately.

\paragraph{The conformal map and the analytic prefactor}

First, we need that $f$ is such that it maps region $[k]$, for $k=1,2,3$, to a subset of the region $S_k$ defined in \eqref{def sector Sk}. If this is true, we can use the asymptotic behavior \eqref{RHP:A-c} of the functions $A_k$ to conclude from \eqref{def P} that
\begin{multline}\label{matching P1}
P(\zeta)=E(z)\left(s^{\frac{2}{3}\rho}f(\zeta)\right)^{-\frac{\sigma_3}{4}}
\begin{pmatrix}1&i\\1&-i\end{pmatrix}
               \left[I+O\left(s^{-\rho}\right)\right]\\ 
              \times\quad 
            e^{-\frac{2s^\rho}{3}f(\zeta)^{3/2}\sigma_3}e^{s^{\rho} q(\zeta) \sigma_3}  G(\zeta)^{-\frac{\sigma_3}{2}}
\end{multline}
for $\zeta\in\partial \mathbb D_{\delta}(b_1)$ as $s\to + \infty$, if $\delta>0$ is sufficiently small.
This has to be $P^{\infty}(\zeta)(1+o(1))$, which suggests us to take $f$ and $E$ as follows,
\begin{align}
&f(\zeta)=\left(\frac{3}{2}q(\zeta)\right)^{2/3},\\
&E(\zeta)= P^\infty(\zeta)  G(\zeta)^{\frac{\sigma_3}{2}} \begin{pmatrix}1&i\\1&-i\end{pmatrix}^{-1} \le(s^{\frac{2}{3} \rho}f(\zeta)\ri)^{\frac{\sigma_3}{4}}.
\end{align}
First, it can be verified using the jump relation for $P^\infty$ and by taking into account the branch cuts of the roots that $E$ is indeed analytic in $\mathbb D_{\delta}(b_1)$.
Secondly, using the properties of the $g$-function, namely \eqref{def g2}, we have as $z \to b_1$,
\begin{equation}
q''(z) = - \frac{(c_1 + c_2)}{2 \sqrt{2}} \frac{\sqrt{\left|\Re(b_1)\right|}}{b_1} (z - b_1)^{-\frac{1}{2}} + \mathcal{O}\left((z - b_1)^{\frac{1}{2}}\right) 
\end{equation}
and hence
\begin{equation}
q(z) = -\frac{2}{3} \frac{(c_1 + c_2)}{\sqrt{2}} \frac{\sqrt{\left|\Re(b_1)\right|}}{b_1} (z - b_1)^{\frac{3}{2}} + \mathcal{O}\left((z - b_1)^{\frac{5}{2}}\right).
\end{equation}
It follows that $f$ is indeed a conformal map and $f'(b_1) \in \C$ with
\begin{equation} \label{eq: der_conformalmap_b1}
\text{arg}\left[ f'(b_1)\right]=\frac{2\phi}{3} \in \left[0, \frac{\pi}{3}\right), \ \text{since } \phi \in \le[0, \frac{\pi}{2}\ri).
\end{equation}

We can now verify that $f$ maps the regions $[1], [2], [3]$ from Figure \ref{figure: Airy} to the admissible sectors $S_1, S_2, S_3$ in \eqref{def sector Sk}. From \eqref{eq: der_conformalmap_b1}, it follows indeed that region $[1]$, where $-\phi<\arg\le[ \zeta-b_1 \ri] <\pi-\phi-\epsilon$, is mapped into the sector $S_1$, that region $[2]$, defined by $\pi-\phi-\epsilon<\arg\le[ \zeta-b_1\ri] <\pi+\epsilon$, is mapped into the sector $S_2$, and that region $[3]$, where $-\pi+\epsilon<\arg\le[ \zeta-b_1\ri]<-\phi$, is mapped into the sector $S_3$.

\subsection{Final transformation $S \mapsto R$}

For the final transformation we define 
\begin{equation} \label{eq: transfSR}
R(\zeta) = S(\zeta) \begin{dcases}\le(P(\zeta)\ri)^{-1} & \text{if } \zeta \in \mathbb{D}_{\delta}(b_1) \cup \mathbb{D}_{\delta}(b_2) \\ \le(P^{\infty}(\zeta)\ri)^{-1} & \text{elsewhere}\end{dcases}.
\end{equation}

It follows that $R(\zeta)$ satisfies the following RH problem:

\subsubsection*{RH problem for $R$} 
\begin{enumerate}
\item $R$ is analytic in $\C \setminus \Gamma_R$ (see \figurename \ \ref{RHR} for the definition of the contour $\Gamma_R$);
\item $R_+(\zeta) = R_-(\zeta) J_R(\zeta)$ for $\zeta \in \Gamma_R$ with 
\begin{equation}
J_R(\zeta) = \begin{dcases}P_-^{\infty}(\zeta) J_S(\zeta) \le(P_+^{\infty}(\zeta)\ri)^{-1} & \text{ if } \zeta \in \Gamma_R \setminus (\partial\mathbb{D}_{\delta}(b_1) \cup \partial \mathbb{D}_{\delta}(b_2)) \\ P(\zeta) \le(P^\infty(\zeta)\ri)^{-1} & \text{ if } \zeta  \in \partial \mathbb{D}_{\delta}(b_1) \cup \partial \mathbb{D}_{\delta}(b_2),\end{dcases}
\end{equation}
 where $J_S(\zeta)$ is given by \eqref{eq: jump_S}, and where we choose the clockwise orientation for the circles around $b_1, b_2$;
\item as $\zeta \to \infty$ 
\begin{equation}
\displaystyle R(\zeta) = I + \frac{R_1(s)}{\zeta} + \mathcal{O}\le(\frac{1}{\zeta^2}\ri). \label{Rinfty}
\end{equation}
\end{enumerate}
By construction, the jump matrix for $R$ is close to the identity matrix as $s \to + \infty$ uniformly in $\zeta$. We have
\begin{equation}
J_R(\zeta) = \begin{dcases}I + \frac{J_R^{(1)}(\zeta)}{s^\rho}+ \mathcal{O}\le(s^{-2\rho}\ri) & \text{ if } \zeta \in \partial \mathbb{D}_{\delta}(b_1) \cup \partial \mathbb{D}_{\delta}(b_2), \\ I + \mathcal{O}\le(e^{- c s^{\rho}}\ri) & \text{ elsewhere,}\end{dcases} \label{jumpforR}
\end{equation}
for some fixed constant $c > 0$ and for some function $J_R^{(1)}(\zeta)$ independent of $s$. Hence, by standard arguments for small norm RH problems (see for example \cite[Section 5.1.3]{smallnormRH}), it follows that
\begin{equation}\label{expansion R s}
R(\zeta) = I + \frac{R^{(1)}(\zeta)}{s^\rho}+\mathcal{O}\le(s^{-2\rho}\ri) \qquad  \text{as } s \to + \infty,
\end{equation}
uniformly for $\zeta\in \C\setminus\Gamma_R$. The error term $R^{(1)}(\zeta)$ can be computed explicitly in terms of $J_R^{(1)}$ (see \cite{DKMcLVZ3}), and
as $\zeta\to\infty$ it behaves like
\begin{equation}\label{expansion R1 zeta}R^{(1)}(\zeta)=\frac{R_1^{(1)}}{\zeta}+\mathcal{O}\le(\frac{1}{\zeta^{2}}\ri),
\end{equation}
for some constant matrix $R_1^{(1)}$. In conclusion, the asymptotic value of $R_1(s)$ defined in (\ref{Rinfty}) is equal to
\begin{equation}\label{expansion R1 s}
R_1(s) =  \frac{R_1^{(1)}}{s^\rho}+\mathcal{O}\le(s^{-2\rho}\ri) \qquad  \text{as } s \to + \infty.
\end{equation}

\begin{figure}
\centering
\scalebox{1.7}{
\begin{tikzpicture}[>=stealth]
\path (0,0) coordinate (O);

\draw[->- = .7] (-2.7,1.7) -- (-1.23,0.18);
\draw[->- = .7] (1.23,0.18) -- (2.7,1.7);

 \draw[->- = .7]  (-0.75,-.15) -- (0,-.5);
  \draw (0,-.5) -- (0.75,-.15);

\node [above] at (-1.2,.5) {\tiny $\Sigma_1$};
\node [above] at (-2,-.2) {\tiny $\Sigma_3$};
\node [above] at (1.2,.5) {\tiny $\Sigma_2$};
\node [above] at (2,-.2) {\tiny $\Sigma_4$};
\node [below] at (0,0.1) {\tiny$\Sigma_5$};

\node [below] at (-1,-.4) {\tiny $\mathbb{D}_\delta(b_1)$};
\node [below] at (1,-.4) {\tiny $\mathbb{D}_\delta(b_2)$};

\draw[->- = .7]   (-2.8,-.4) -- (-1.3,-.1);
\draw[->- = .7]  (1.3,-.1) -- (2.8,-.4);

\draw (-1,0) circle (.3cm);
\draw (1,0) circle (.3cm);

\end{tikzpicture}}
\caption{The contour $\Gamma_R$ for the RH problem for $R(\zeta)$.}
\label{RHR}
\end{figure}
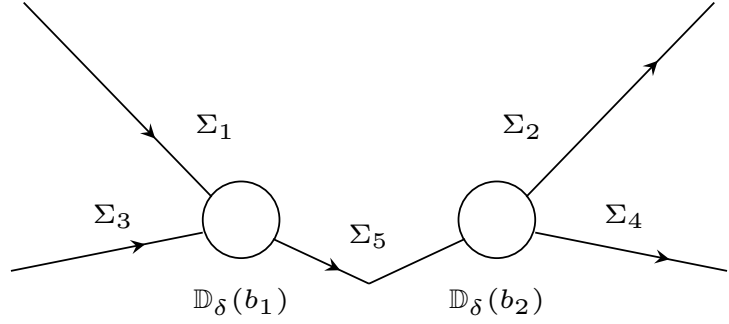

\begin{remark}
\label{remark: phinegative}
In the RH analysis, we restricted ourselves to the cases $j=1, 2$ and $j=3$ with $\theta\leq 1$. If $j=3$ and $\theta>1$, some modifications are required.  Since the angle $\phi$ in \eqref{def b1b2} becomes negative, the endpoints $b_1, b_2$ lie in the lower half plane and the jump contours in the RH problems for $T$ and $S$ need to be changed. Compared to the contours in Figure \ref{fig: TomRHPstraight}, all contours have to be mirrored with respect to the real line, in particular we need to take
\[\Sigma_5\mapsto \overline{\Sigma_5},\ \Sigma_1\mapsto\overline{\Sigma_3},\ \Sigma_2\mapsto\overline{\Sigma_4},\ \Sigma_3\mapsto\overline{\Sigma_1},\ \Sigma_4\mapsto\overline{\Sigma_2}.\]
With these modified contours, it is straightforward to adapt the proof of Lemma \ref{lemma3.1} to the case where $\phi<0$. The rest of the analysis, namely the construction of the global and local parametrices and of the small norm RH problem for $R$, is similar to the case $\phi\geq 0$.
\end{remark}

\section{The Fredholm determinant} \label{sec: fredholm_det}

We can now invert all the transformations $Y \mapsto U \mapsto T \mapsto S \mapsto R$ and find an explicit asymptotic expression for the Fredholm determinant as $s\to +\infty$. From (\ref{identity Fredholm U}), we know that
\begin{equation*}
\frac{\d}{\d s} \ln \det\le(1 - \K^{(j)}\bigg|_{[0, s]}\ri) = - i s^{\rho^{(j)} - 1} \le(U_1(s)\ri)_{2, 2}.
\end{equation*}
 Combining \eqref{S1}, \eqref{Pinftyinfty}, \eqref{eq: Pinfty} and \eqref{Rinfty}, we obtain
 \begin{align}
\frac{\d}{\d s} \ln \det\le(1 - \K^{(j)}\bigg|_{[0, s]}\ri) &= - i s^{\rho - 1} \le(\le(S_1(s)\ri)_{2,2} - s^{\rho} g_1\ri) \nonumber \\
&=   i g_1 s^{2 \rho - 1} - i s^{\rho - 1} \le( \le(P^{\infty}_1(s)\ri)_{2,2} + \le(R_1(s) \ri)_{2,2}  \ri) \nonumber \\
&=   i g_1 s^{2 \rho - 1} - i s^{\rho - 1} \le( -p_1(s) + \le(R_1(s) \ri)_{2,2}  \ri). \label{Fredholm expansion 1}
\end{align}
The large $s$ asymptotics for $\le(R_1(s) \ri)_{2,2}$ are given in \eqref{expansion R1 s}, and we will now compute $g_1(s)$ and $p_1(s)$.

\subsection{Calculation of $g_1$}
We recall the definition of the second derivative of the $g$-function (see (\ref{def g2})),
\begin{equation}\label{def g2 bis}
g''(\zeta)=-i\frac{c_1+c_2}{2}\left(\frac{1}{\zeta}-\frac{1}{r(\zeta)}+\frac{i \Im b_1}{\zeta r(\zeta)}\right),
\end{equation}
where as before
\begin{equation}
r(\zeta):=\le[(\zeta-b_1)(\zeta-b_2)\ri]^{\frac{1}{2}}.
\end{equation}

In order to calculate $g_1$, we expand (\ref{def g2 bis}) as $\zeta \to \infty$ and obtain
\begin{align}
g''(\zeta) &= -\frac{i(c_1+c_2)}{4\zeta^3} \le[ b_1b_2 + \le(\Im(b_1)\ri)^2 \ri] + \mathcal{O}\le(\frac{1}{\zeta^4}\ri) \nonumber \\
&= \frac{i (c_1+c_2) \le(\Re(b_1)\ri)^2 }{4 \zeta^3} + \mathcal{O}\le(\frac{1}{\zeta^4}\ri).
\end{align}
On the other hand, $g(\zeta) := \frac{g_1}{\zeta} + \mathcal{O}\le(\zeta^{-2}\ri)$ as $\zeta \rightarrow \infty$ (see \eqref{eq: cond2_gfun}), which implies 
\begin{equation}
g''(\zeta) = \frac{2g_1}{\zeta^3} + \mathcal{O}\le(\frac{1}{\zeta^4}\ri)
\end{equation}
Therefore, we can identify the coefficient $g_1$ as
\begin{equation}
g_1 = \frac{i \le(\Re(b_1)\ri)^2 (c_1+c_2)}{8 }. \label{g1explicit}
\end{equation}

\subsection{Calculation of $p_1(s)$}
We recall from \eqref{def p1} that
\begin{equation}
p_1 (s)= \frac{1}{2\pi i} \int_{\Sigma_5} \frac{(\xi-i\Im (b_1)) \ln  G(\xi)}{r_+(\xi)} \d\xi.
\end{equation}
We can split the integration over $\Sigma_5 = [b_1,0] \cup [0,b_2]$ into integration over $C_1$ and $C_2$ with
\[
C_1 = \Sigma_5 \cap \le\{ \le|\zeta  \ri| < M s^{-\rho} \ri\},\qquad
C_2= \Sigma_5 \setminus C_1,
\]
for some fixed sufficiently large $M>0$.  

By (\ref{lnGbdd}), one checks that the contribution of the integrals over $C_1$ to $p_1$ can be written as  $\frac{a_1(M)}{s^\rho}+\mathcal{O}\le(s^{-2\rho}\ri)$ as $s\to +\infty$, for some constant $a_1(M)$ depending on $M$. On the other hand, for the integral over $C_2$ we can use (\ref{as G}), and obtain
\begin{multline}
p_1 (s)=  \int_{C_2} \frac{(\xi-i\Im (b_1)) \le[ c_4 \ln s + c_5 \ln \le( i \xi \ri)+ c_6 \ln \le( -i\xi\ri)+ c_7 \ri]}{r_+(\xi)}  \frac{\d\xi}{2\pi i} \\ 
\qquad + s^{-\rho} \le[  - \frac{c_8}{2\pi } \int_{C_2} \frac{(\xi - i\Im (b_1))\d \xi}{\xi\, r_+(\xi)} +a_1(M)\ri] + \mathcal{O}\le( \frac{1}{s^{2\rho}}\ri) 
\end{multline}
as $s \to + \infty$. We now replace the first two integrals over $C_2$ again by integrals over the whole contour $\Sigma_5$; this implies adding a contribution of order $s^{-\rho}$ (possibly depending on $M$) which will be counted in the $s^{-\rho}$ term in the formula below. We get

\begin{align}
p_1(s) &=  (c_4\ln s + c_7) \int_{\Sigma_5} \frac{\xi -i\Im (b_1)}{r_+(\xi)} \frac{\d\xi}{2\pi i} \nonumber \\
& \qquad+ c_5 \int_{\Sigma_5} \frac{(\xi -i\Im (b_1)) \ln \le( i \xi \ri)}{r_+(\xi)} \frac{\d\xi}{2\pi i} + c_6 \int_{\Sigma_5} \frac{(\xi -i\Im (b_1))  \ln \le( -i\xi\ri) }{r_+(\xi)} \frac{\d\xi}{2\pi i}  \nonumber  \\
& \qquad + s^{-\rho} \le[  - \frac{c_8}{2\pi } \int_{C_2} \frac{(\xi - i\Im (b_1))\d \xi}{\xi\, r_+(\xi)} +a_1(M) + a_2(M)\ri] + \mathcal{O}\le( \frac{1}{s^{2\rho}}\ri)  \nonumber \\
& =:  \le(c_4 \ln s + c_7\ri) I_1 + c_5 I_2 + c_6  I_3 +  \frac{\mathcal{K}}{s^{\rho}} + \mathcal{O}\le( \frac{1}{s^{2\rho}}\ri),\label{p1Ij}
\end{align}
for some constant $\mathcal K$. The value of $\mathcal K$ depends on the parameters $\{\nu_j \}$, $\{\mu_k\}$, $\alpha$, $\theta$ but we do not compute its explicit value.
Note that $\mathcal K$ does not depend on $M$, although it may seem to a priori, since $p_1(s)$ does not depend on $M$.
The integrals $I_1,I_2, I_3$ are defined as
\begin{align*}
&I_1 = \int_{\Sigma_5} \frac{\xi -i\Im (b_1)}{r_+(\xi)} \frac{\d\xi}{2\pi i},\\
&I_2 = \int_{\Sigma_5} \frac{(\xi -i\Im (b_1)) \ln \le( i \xi \ri)}{r_+(\xi)} \frac{\d \xi}{2 \pi i},&&
I_3 = \int_{\Sigma_5} \frac{(\xi -i\Im (b_1)) \ln \le( -i \xi \ri)}{r_+(\xi)} \frac{\d \xi}{2 \pi i},
\end{align*}
and they remain to be computed.

\underline{Computation of $I_1$.}
We assume that the endpoints $b_1, b_2$ lie in the upper half-plane, as set in Section \ref{RHPforU} (for the other case, the argument is similar).
By analyticity we can deform the contour $\Sigma_5$ to a horizontal segment between $b_1$ and $b_2$ and we can easily show that $I_1$ is zero: 
\begin{equation}
I_1 =  \int_{\Sigma_5} \frac{(\xi - i\Im (b_1))\d\xi}{2\pi ir_+(\xi)}=-\int_{\Re (b_1)}^{\Re (b_2)} \frac{u}{\sqrt{\Re(b_1)^2-u^2}} \frac{\d u}{2\pi}=0 \label{I1}
\end{equation}
by symmetry.

\underline{Computation of $I_3$.} 
For the integral $I_3$, we use again analyticity to deform as before the contour $\Sigma_5$ into the segment $[b_1, b_2]$ (the logarithmic branch cut lies on $i\R^-$): we obtain
\begin{multline}
I_3 = \int_{\Sigma_5} \frac{(\xi -i\Im (b_1))\ln \le( -i \xi \ri)}{r_+(\xi)} \frac{\d \xi}{2\pi i }  
 = -\int_{\Re(b_1)}^{\Re(b_2)} \frac{u\ln \le(-iu +\Im(b_1) \ri)}{\sqrt{\le(\Re(b_1)\ri)^2 - u^2}} \frac{\d u}{2\pi } \\ = -\frac{i}{\pi} \int_{0}^{\Re(b_2)} \frac{ u\arg\le[ -iu + \Im(b_1) \ri]}{ \sqrt{\le(\Re(b_1)\ri)^2 - u^2}} \d u .
\end{multline}

The last integral is equal to $-\frac{\pi}{2}(|b_1|-\Im(b_1))$. Therefore, 
\begin{equation}
I_3 = \frac{i}{2}(|b_1|-\Im(b_1)). \label{I3}
\end{equation}

\underline{Computation of $I_2$.}
For the integral $I_2$, the branch cut of the logarithm is on $i\R^+$, and the integration over $\Sigma_5$ can be deformed by analyticity to a contour as showed in \figurename \ \ref{contourI2}.

\begin{figure}[h!]
\centering
\begin{tikzpicture}[>=stealth]
\path (0,0) coordinate (O);

\draw[dashed] (0,0) -- (0,3) coordinate (y axis);

\draw[->- = .4]  (-3,1) -- (-0.05,1); 
\draw[->- = .6]  (0.05,1) -- (3,1);

\draw[->- = .4]  (-0.05,1) -- (-0.05,0); 
\draw[->- = .4]  (0.05,0) -- (0.05,1);

\draw[fill] (0,0) circle [radius=0.03];

\draw[fill] (-3,1) circle [radius=0.03];
\draw[fill] (3,1) circle [radius=0.03];

\node [above left] at (-3,1) {\scriptsize $b_1$};
\node [above right] at (3,1) {\scriptsize $b_2$};

\node [below] at (0,0) {\scriptsize $0$};

\end{tikzpicture}
       \caption{Deformation of the contour for the integral $I_2$.}
\label{contourI2}
\end{figure}
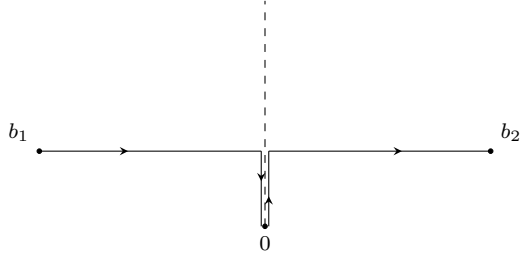

Given a parametrization of the form $\xi = i \Im(b_1) + u$, $u \in [\Re(b_1), \Re(b_2)]$ for the horizontal parts and  the parametrization $\xi = iv$, $v \in [0, \Im(b_1)] $ for the vertical parts, we have
\begin{align}
I_2 &= \int_{\Sigma_5} \frac{(\xi -i\Im (b_1))\ln \le( i \xi \ri)}{r_+(\xi)} \frac{\d \xi}{2\pi i } \nonumber \\
& = -\int_{\Re(b_1)}^{\Re(b_2)} \frac{u\ln \le(iu -\Im(b_1) \ri)}{\sqrt{\le(\Re(b_1)\ri)^2 - u^2}} \frac{\d u}{2\pi } + \int_0^{\Im(b_1)} \frac{(v - \Im (b_1))(\ln_+(-v) - \ln_-(-v))}{ \sqrt{\le(\Re(b_1)\ri)^2 + \le(v-\Im(b_1)\ri)^2} }\frac{\d v}{2\pi} \nonumber \\
& = -\frac{i}{\pi} \int_{0}^{\Re(b_2)} \frac{ u\arg\le[ iu - \Im(b_1) \ri]}{ \sqrt{\le(\Re(b_1)\ri)^2 - u^2}} \d u+ i \int_0^{\Im(b_1)} \frac{ v-\Im(b_1) }{ \sqrt{\le(\Re(b_1)\ri)^2 + \le(v-\Im(b_1)\ri)^2} }\d v.
\end{align}
The second integral is equal to $-|b_1|-\Re(b_1)$.
For the first integral, we note that
\[\arg\le[ iu - \Im(b_1) \ri]=\arg\le[ -iu + \Im(b_1) \ri]+\sgn (u) \pi,\]
and then it follows that
\begin{equation}\label{I2 final}
I_2=I_3-i|b_1|= \frac{i}{2}(|b_1|-\Im(b_1))-i|b_1|.
\end{equation}

In conclusion, substituting (\ref{I1}), (\ref{I2 final}) and (\ref{I3}) (\ref{p1Ij}), we get as $s \to + \infty$,
\begin{equation}
p_1(s) =- i c_5 |b_1|  +i \frac{c_5+c_6}{2 }   \le(|b_1|-\Im(b_1) \ri) +  \frac{\mathcal{K}}{s^{\rho}} + \mathcal{O}\le( \frac{1}{s^{2\rho}}\ri) \label{p1explicit}
\end{equation}
as $s \to + \infty$.

\subsection{The final asymptotic expansion of the Fredholm determinant}\label{subsec: log term}

Using \eqref{expansion R1 s}, \eqref{g1explicit} and \eqref{p1explicit} in \eqref{Fredholm expansion 1}, we obtain
\begin{multline}
\frac{\d}{\d s} \ln \det\le(1 - \K^{(j)}\bigg|_{[0, s]}\ri) =  - \frac{ \le(\Re(b_1)\ri)^2 (c_1+c_2)}{8 } s^{2 \rho - 1}   \\
 -  \le( -c_5 |b_1|  + \frac{c_5+c_6}{2}   \le(|b_1|-\Im(b_1)  \ri) \ri) s^{\rho - 1}  \nonumber \\
+ \frac{-\mathcal{K} + \le(R_1^{(1)}\ri)_{2,2} }{is}  + \mathcal{O}\le( s^{-\rho-1}\ri),\qquad s\to + \infty .
\end{multline}
Integrating in $s$, we obtain
\begin{equation}
\ln \det\le(1 - \K^{(j)}\bigg|_{[0, s]}\ri)
 =  - a^{(j)} \, s^{2 \rho}   
+b^{(j)} \, s^{\rho}
 + c^{(j)}\, \ln s  + \ln C^{(j)} +\mathcal{O}\le( s^{-\rho}\ri),\qquad s\to + \infty,
\end{equation}
for some integration constant $\ln C^{(j)}$, and with
\begin{align}
&a^{(j)}=\frac{ \le(\Re(b_1)\ri)^2 (c_1+c_2)}{16\rho },\\
&b^{(j)}=-\frac{1}{\rho}  \le(- c_5 |b_1|  + \frac{c_5+c_6}{2}   \le(|b_1|-\Im(b_1) \ri) \ri).
\end{align}
Combining \eqref{eq: reb2_j12} and \eqref{eq: reb2_j3} with the specific values for the constants $\{c_i\}$ from \eqref{eq: ci_j1}, \eqref{eq: ci_j2} and \eqref{eq: ci_j3} we immediately get Theorem \ref{theorem large gap}.

\begin{remark}
The values of $\mathcal{K}$ and $\le(R^{(1)}_{1}\ri)_{2,2}$ will determine the coefficient $c^{(j)}$ in front of the logarithmic term in the asymptotic expansion of the Fredholm determinant. As already stressed in the introduction, their value can in principle be explicitly computed, but the computations are quite involved, and we do not proceed with this. 
\end{remark}

\appendix

\section{Limit of the smallest eigenvalue distribution as a Fredholm determinant}
\label{appendix: fredholm}
\paragraph{Correlation kernels and scaling limits}
We first express the finite $n$ correlation kernels $K_n^{(j)}$ as double contour integrals.
\begin{proposition}\label{prop: finite n kernels} We denote $K_n^{(1)}$ for the eigenvalue correlation kernel of $M^{(1)}$, $K_n^{(2)}$ for the eigenvalue correlation kernel of $M^{(2)}$, and $K_n^{(3)}$ for the correlation kernel of the determinantal point process \eqref{MB}. The correlation kernels admit the following double integral representations:
\begin{equation}\label{correlation kernel Kn}
 K_n^{(j)}(x,y)=\frac{1}{4\pi^2}\int_{\gamma} \d u\int_{\tilde\gamma} \d v \frac{F_n^{(j)}(u)}{F_n^{(j)}(v)}\frac{x^{-u}y^{v-1}}{u-v},\qquad j=1,2,3,
 \end{equation}
 with contours $\gamma$ and $\tilde \gamma$ as shown in \figurename \ \ref{Kalphatheta}.
 The functions $F_n^{(1)}, F_n^{(2)}, F_n^{(3)}$ are given by
\begin{align}
&\label{Fn1}F_n^{(1)}(z)=\frac{\Gamma(-z-n+1)}{\prod_{k=0}^r\Gamma(-z+\nu_k+1)}\\
&\label{Fn2}F_n^{(2)}(z)=\prod_{k=0}^r\frac{ \Gamma\le(1+\ell_k-n -z\ri)}{  \Gamma\le(1+\nu_k -z\ri)}\\
&\label{Fn3}F_n^{(3)}(z)=\frac{\Gamma(z+\frac{\alpha}{2})}{\Gamma(\frac{\frac{\alpha}{2}+1-u}{\theta})}\Gamma\left(n+\frac{\frac{\alpha}{2}+1-u}{\theta}\right).
\end{align}  
Moreover, we have for $j=1,2,3$ and $s>0$,
\begin{equation}\label{norm operators}
\int_0^s\int_0^s \left|K_n^{(j)}(x,y)\right|^2 \d x\d y<\infty.
\end{equation}
\end{proposition}
\begin{proof}
\begin{enumerate}
\item For $j=1$, it was shown in  \cite[formula (5.1)]{ArnoLun} that
\begin{multline}
K_n^{(1)}(x,y) = \frac{1}{(2\pi i)^2} \int_{-\frac{1}{2}+i \R} \d s \int_{\Sigma_n} \d t  \frac{x^t y^{-s-1}}{s-t}\\
\times\ \prod_{j=0}^r \frac{\Gamma(s+\nu_j+1)}{\Gamma(t+\nu_j+1)} \frac{\Gamma(t-n+1)}{\Gamma(s-n+1)},\label{discrete kernel j=1 AL}
\end{multline}
with $\Sigma_n$ a contour enclosing $0,1,\ldots, n-1$ once in the counterclockwise direction and such that ${\Re}\, t>-\frac{1}{2}$.
Substituting $u=-t$ and $v=-s$, we obtain a double contour integral representation of the kernel $K^{(1)}_n$ with curves $-\Sigma_n$ and $\frac{1}{2}+i \R$. Since $F^{(1)}$ is analytic off the real line, and decays/blows up fast in the left/right half plane (this follows from Stirling's approximation after a straightforward calculation), we can deform $-\Sigma_n$ to $\gamma$ and $\frac{1}{2}+i \R$ to $\tilde\gamma$ without modifying the integrals and thus obtaining \eqref{correlation kernel Kn}.

The function $F_n^{(1)}$ has poles at $-n+1, -n+2,\ldots, 0$ and possibly zeros at $1,2,\ldots$
One can choose $\gamma$ and $\tilde\gamma$ arbitrarily close to the points $0$ and $1$, in such a way that they cross the real line at points $\epsilon$ and $1-\epsilon$. From \eqref{correlation kernel Kn}, we then obtain immediately that $K_n^{(1)}(x,y)=\mathcal O((xy)^{-\epsilon})$ as $x,y\to 0_+$, which implies \eqref{norm operators}.

\item For $j=2$, it was shown in \cite[formula (2.33)]{ArnoDriesMario} that
\begin{multline}
K_n^{(2)}(x,y) = \frac{1}{(2\pi i)^2} \int_C \d s \int_{\Sigma_n} \d t  \frac{ x^t y^{-s-1} }{s-t}\\ \times\ \prod_{j=0}^r \frac{ \Gamma(s+1+\nu_j) \Gamma(t+1+\ell_j-n) }{ \Gamma(t+1+\nu_j) \Gamma(s+1+\ell_j-n) }, \label{discrete kernel j=2 AMD}
\end{multline}
where $\Sigma_n$ is as in the case $j=1$ and $C$ is a counterclockwise oriented curve which starts and ends at $-\infty$ and encircles the negative real line. 

Again, substituting $u=-t$ and $v=-s$ and deforming the integration contours (thanks to the analyticity of the functions involved), we obtain \eqref{correlation kernel Kn}.
Similarly as for $j=1$, we have $K_n^{(2)}(x,y)=\mathcal O((xy)^{-\epsilon})$ as $x,y\to 0_+$, which implies \eqref{norm operators}.
\item For $j=3$, it was shown in \cite[formula (1.11)]{LZhang} that the correlation kernel $K_n^{(3)}$ can be written as
\begin{multline}
K_n^{(3)}(x,y) = \frac{\theta}{(2\pi i)^2} \int_{c+i \R} \d s \int_{\Sigma_n} \d t \frac{x^{-\theta s-1}y^{\theta t}}{s-t} \\\times\  \frac{\Gamma(s+1)\Gamma(\alpha+1+\theta s)\Gamma(t-n+1)}{\Gamma(t+1)\Gamma(\alpha+1+\theta t)\Gamma(s-n+1)}
\end{multline}
with $c=-\frac{1}{2}+\frac{1}{2}\max\{0,1-\frac{\alpha+1}{\theta}\}$, and $\Sigma_n$ a closed counter-clockwise contour going around $0,1,\ldots, n-1$ and such that ${\rm Re }\, t>c$.
However, in order to have a correlation kernel $K_n^{(3)}$ satisfying \eqref{norm operators}, we multiply $K_n$ with the gauge factor $x^{\alpha/2}y^{-\alpha/2}$. This preserves the associated determinantal point process. After substituting $u=-\theta t-\frac{\alpha}{2}-1$ and $v=-\theta s-\frac{\alpha}{2}-1$ and appropriately deforming the integration contours, we obtain \eqref{correlation kernel Kn} with $F^{(3)}$ given by \eqref{Fn3}.

One can choose $\gamma$ and $\tilde\gamma$ arbitrarily close to the points $-\frac{\alpha}{2}$ and $1+\frac{\alpha}{2}$, in such a way that they cross the real line at points $\frac{\alpha}{2}+\epsilon$ and $1+\frac{\alpha}{2}-\epsilon$. From \eqref{correlation kernel Kn}, we then obtain immediately that $K_n^{(3)}(x,y)=\mathcal O((xy)^{\alpha/2-\epsilon})$ as $x,y \to 0_+$. Provided that $0<\epsilon < \frac{\alpha+1}{2}$ (this ensures that $\frac{\alpha}{2}-\epsilon>-\frac{1}{2}$), we have \eqref{norm operators}.
\end{enumerate}
\end{proof}
As a consequence of \eqref{norm operators}, the integral operators $\left.K_n^{(j)}\right|_{[0,s]}$
defined by 
\begin{equation}\label{def operators}
\left.K_n\right|_{[0,s]} f(y)=\int_0^s K_n(x,y)f(x)\d x,\qquad f\in L^2(0,s),\quad y\in [0,s]
\end{equation}
are well-defined bounded linear operators on $L^2(0,s)$.
They are of finite rank $n$ and thus trace-class, hence the Fredholm determinants $\det\left(1-\left. K_n^{(j)}\right|_{[0,s]}\right)$
are well-defined.

\begin{proposition}\label{prop: limiting kernels}Let $c_n^{(1)}, c_n^{(2)}, c_n^{3)}$ be defined by \eqref{def cj}.
The scaling limits \eqref{scaling limit} hold for $j=1,2,3$, and the limiting kernels $\mathbb K^{(1)}, \mathbb K^{(2)}, \mathbb K^{(3)}$ are given by \eqref{limiting kernels1}. 
\end{proposition}
\begin{proof}
\begin{enumerate}
\item For $j=1$, the scaling limit was proven in \cite[Theorem 5.3]{ArnoLun}, where the discrete kernel $K_n^{(1)}$ had the form \eqref{discrete kernel j=1 AL} and the limiting kernel was
\begin{gather}
\mathbb{K}^{(1)}(x,y) =  \frac{1}{(2\pi i)^2} \int_{-\frac{1}{2}+i\R} \d s \int_{\Sigma} \d t  \frac{ x^t y^{-s-1} }{s-t}  \prod_{j=0}^r \frac{ \Gamma(s+1+\nu_j) }{ \Gamma(t+1+\nu_j) } \frac{\sin(\pi s)}{\sin (\pi t)} 
\end{gather}
with $\Sigma$ a contour around the positive real axis in the half-plane $\Re t > - \frac{1}{2}$. We now use Euler's reflection formula for the $\Gamma$ function $\Gamma(z)\Gamma(1-z) = \frac{\pi}{\sin (\pi z)}$ and we recall the fact that $\nu_0=0$. By a standard change of variables ($u=-t$ and $v=-s$) and  deformation of contours, we obtain $\mathbb{K}^{(1)}$ as given in \eqref{limiting kernels1}.
\item Similarly, for $j=2$, the scaling limit was shown in \cite[Theorem 2.8]{ArnoDriesMario}, where the finite $n$ kernel $K_n^{(2)}$ has the form \eqref{discrete kernel j=2 AMD} and the limiting kernel is
\begin{gather}
\mathbb{K}^{(2)}(x,y) = \nonumber \\ \frac{1}{(2\pi i)^2} \int_{-\frac{1}{2}+i\R} \d s \int_{\Sigma} \d t  \frac{ x^t y^{-s-1} }{s-t}  \prod_{j=0}^r \frac{ \Gamma(s+1+\nu_j) }{ \Gamma(t+1+\nu_j) } \frac{\sin(\pi s)}{\sin (\pi t)} \prod_{k\in J} \frac{\Gamma(t+1+\ell_k)}{\Gamma(s+1+\ell_k) }.
\end{gather}
As for the case $j=1$, similar straightforward manipulations lead to  the expression for the kernel $\mathbb{K}^{(2)}$ as given in \eqref{limiting kernels1}.
\item We recall the definition of the limiting kernel $\mathbb K^{(3)}$ appearing in \cite{Borodin} (see also \cite{TomStefano}),
\begin{gather}
\mathbb K^{(3)}(x,y) 
= \theta \le(xy\ri)^{\frac{\alpha}{2}} \int_0^1 J_{\frac{\alpha+1}{\theta}, \frac{1}{\theta}}\le(xt\ri) \, J_{\alpha+1, \theta} \le( (yt)^\theta \ri) t^\alpha \d t \label{Borodin01}
\end{gather}
where $J_{a,b}(x) = \sum_{m=0}^\infty \frac{(-x)^m}{m! \Gamma \le( a+bm \ri)}$ is the Wright's generalized Bessel function. The scaling limit result has been proved in \cite[Theorem 4.2]{Borodin}.

By a residue calculation, it is easy to see that the Wright's Bessel function can be expressed as
\begin{gather}
J_{a,b}(x) = \int_\gamma \frac{\d u}{2\pi i} x^{-u} \frac{\Gamma(u)}{\Gamma\le(a-bu\ri)}
\end{gather}
where $\gamma$ can be any curve that encloses all the negative integers and the origin in the counterclockwise direction. 

Substituting the above expression in the definition of $\mathbb{K}^{(3)}$ and performing some integrations and changes of variables, one easily obtains the desired representation of the kernel \eqref{limiting kernels1}. An equivalent double-contour representation has also been obtained in \cite[Corollary 1.2]{LZhang2}.
\end{enumerate}

\end{proof}

We will now justify the limit \eqref{lim eig distr Fredholm}. To this end, we need to show that the operator acting on $L^2(0,s)$ with kernel $\frac{1}{c_n^{(j)}}K_n^{(j)}\left(\frac{x}{c_n^{(j)}},\frac{y}{c_n^{(j)}}\right)$ converges, as $n\to + \infty$, to the operator with kernel $\mathbb K^{(j)}$ for the trace norm. 
\begin{lemma}\label{lemma scaling limit uniform}
Let $s>0$ and let $K_n^{(j)}$ be the correlation kernels defined in \eqref{correlation kernel Kn}.
For $j=1,2,3$, there exists a constant $\beta\in \le[0,\frac{1}{2}\ri)$ such that
\begin{equation}\label{scaling limit uniform}
\lim_{n\to + \infty}\frac{(xy)^{\beta}}{c_n^{(j)}}K_n^{(j)}\left(\frac{x}{c_n^{(j)}},\frac{y}{c_n^{(j)}}\right)=(xy)^{\beta}\mathbb K^{(j)}(x,y),
\end{equation}
uniformly for $x,y\in [0,s]$. For $x=0$ and $y=0$, both the left and right hand sides of the above equation are understood as the limit as $x,y\to 0_+$. We can take $\beta=\frac{1}{4}$ for $j=1,2$ and $\beta=\max\{-\frac{\alpha}{2},0\}$ for $j=3$.
\end{lemma}
\begin{proof}
\begin{enumerate}
\item The uniform convergence of the kernel $K_n^{(1)}$ has already been proven in \cite[Theorem 5.3]{ArnoLun} for compact subsets of the positive real line. We will perform here almost the same calculations, adding the fact that in our case $x,y \in [0,s]$: the additional factor $(xy)^{\frac{1}{4}}$ will guarantee a uniform convergence also in such a neighbourhood of zero.
\begin{gather}
\frac{ (xy)^{\frac{1}{4}} }{n} K_n^{(1)}\le( \frac{x}{n},\frac{y}{n}\ri) = \nonumber \\
=\frac{(xy)^{\frac{1}{4}}}{4\pi^2}\int_{\gamma}\int_{\tilde\gamma}  \prod_{k=0}^r\frac{\Gamma(1+\nu_k-v)}{\Gamma(1+\nu_k-u)} \frac{\Gamma(1-n-u)}{\Gamma(1-n-v)} \frac{ \le(\frac{x}{n}\ri)^{-u} \le(\frac{y}{n}\ri)^{v-1} }{n(u-v)} \d v \, \d u \nonumber \\
= \frac{1}{4\pi^2}\int_{\gamma}\int_{\tilde\gamma}  \prod_{k=0}^r\frac{\Gamma(1+\nu_k-v)}{\Gamma(1+\nu_k-u)} \frac{\Gamma(1-n-u)}{\Gamma(1-n-v)} \frac{x^{-u+\frac{1}{4}} y^{v-\frac{3}{4}} }{n^{v-u}(u-v)} \d v\, \d u. \label{A.15}
 \end{gather}
Thanks to the analyticity of the contours away from the points $\le\{ 0, -1, \ldots, -n-1 \ri\} \cup \le\{ k+\nu_{\min}\ri\}_{k\in \mathbb N}$, we can assume that $\max \le\{ \Re \le(\gamma(u)\ri) \ri\} < \frac{1}{4}$ and $\min \le\{ \Re \le( \tilde \gamma (v) \ri) \ri\} > \frac{3}{4} $. Thus the factor $x^{-u+\frac{1}{4}} y^{v-\frac{3}{4}}$ is uniformly bounded on $[0,s]$. 

Using the Euler's reflection formula for the Gamma function, we have 
\begin{gather}
 \frac{\Gamma(1-n-u)}{\Gamma(1-n-v)} = \frac{\sin (\pi v) }{\sin (\pi u)} \frac{\Gamma(n+v)}{\Gamma(n+u)};
\end{gather}
moreover, as $n\to + \infty$
\begin{gather}
\frac{\Gamma(n+v)}{\Gamma(n+u)} = n^{v-u} \le(1+ \mathcal{O} \le(\frac{1}{n} \ri)\ri)
\end{gather}
uniformly, thanks to the Stirling formula \cite[formula (6.1.37)]{NIST}. Therefore, the uniform convergence of the integrand holds. 

The integral over $\gamma$ in \eqref{A.15} converges since $\Gamma (1+\nu_j -u)$ increases along the contour $\gamma$ towards $-\infty$ and $\le| \sin(\pi u) \ri| \geq \le| \sinh (\pi) \Im (u) \ri|$. Similarly, $\Gamma (1+\nu_k-v)$ tends to zero at an exponential rate along the contour $\tilde \gamma$ towards $+\infty$ and the integral over $\tilde \gamma$ in \eqref{A.15} converges as well. In conclusion, we can indeed interchange the limit and integrals and obtain the result \eqref{scaling limit uniform}. 
\item Similar arguments hold for the case $j=2$ with $\beta= \frac{1}{4}$.
\item For $j=3$, if $\alpha \geq 0$, the uniform convergence is straightforward, while if $\alpha <0$, we need to introduce again the additional factor $(xy)^\beta$, with $\beta = -\frac{\alpha}{2}$ in order to guarantee uniform convergence in a (right) neighbourhood of zero. The proof is again similar as in the case $j=1$.
\end{enumerate}
\end{proof}
\begin{corollary}\label{corollary Fredholm}
The limit \eqref{lim eig distr Fredholm} holds for $j=1,2,3$.
\end{corollary}
\begin{proof}
Write $U_n^{(j)}$ for the integral operator acting on $L^2(0,s)$ with kernel \[U_n^{(j)}(x,y):=\frac{1}{c_n^{(j)}}K_n^{(j)}\left(\frac{x}{c_n^{(j)}},\frac{y}{c_n^{(j)}}\right), \qquad j=1,2,3.
\]
First, if we take a continuous test function $f$, one shows using Lemma \ref{lemma scaling limit uniform} and the dominated convergence theorem that \[\le\| \le(U_n^{(j)}-\mathbb K^{(j)}\ri) f \ri\|_2\leq \|U_n^{(j)}-\mathbb K^{(j)}\|_2 \ \|f\|_\infty\to 0.\] It follows that $U_n^{(j)} \to \mathbb K^{(j)}$ weakly. Secondly, again using Lemma \ref{lemma scaling limit uniform} and the dominated convergence theorem, we show easily that $\Tr U_n^{(j)}\to \Tr \mathbb K^{(j)}$.
From \cite[Theorem 2.21 and Theorem A.6]{Simon}, it follows that $U_n^{(j)}\to \mathbb K^{(j)}$ in trace norm. Since the Fredholm determinant is continuous under trace norm, we obtain 
\eqref{lim eig distr Fredholm}.
\end{proof}

\section*{Acknowledgements}
TC and MG are supported by the European Research Council under the European Union's Seventh Framework Programme (FP/2007/2013)/ ERC Grant Agreement 307074. DS is supported by FWO Flanders Project G.0934.13
and KU Leuven Research Grant OT/12/073. The authors are also supported by the Belgian Interuniversity Attraction Pole P07/18.

\bibliography{biblionotes}

\end{document}